\documentclass[final,hidelinks,onefignum,onetabnum]{siamart251216}

\usepackage{graphicx}
\ifpdf
  \DeclareGraphicsExtensions{.pdf,.png,.jpg,.eps}
\else
  \DeclareGraphicsExtensions{.eps}
\fi
\usepackage{amssymb,amsfonts}
\usepackage{physics}

\usepackage{tikz}
\usetikzlibrary{arrows.meta, positioning, decorations.markings, calc}
\usepackage{booktabs}
\usepackage{multirow}

\usepackage{enumitem}
\usepackage{tabularx}
\setlist[itemize,2]{
    leftmargin=0pt,
    labelsep=0.5em,
    itemindent=1.5em
}

\usepackage{etoolbox}

\AtBeginEnvironment{flalign}{\setlength{\abovedisplayskip}{1ex}\setlength{\belowdisplayskip}{1ex}}

\usepackage{algpseudocode}
\usepackage{adjustbox}

\newcommand{\D}{\mathrm{d}}

\newsiamremark{remark}{Remark}
\newsiamremark{hypothesis}{Hypothesis}
\crefname{hypothesis}{Hypothesis}{Hypotheses}
\newsiamthm{claim}{Claim}
\newsiamremark{fact}{Fact}
\crefname{fact}{Fact}{Facts}
\newsiamremark{example}{Example}
\newsiamthm{assumption}{Assumption}
\crefname{assumption}{Assumption}{Assumptions}
\crefname{example}{Example}{Examples}

\headers{General Numerical Solver for NMQSD}{Z. Cai and Q. Zhu}

\title{A General First- and Second-Order Numerical Solver for Non-Markovian Quantum State Diffusion}
  
\author{
Zhenning Cai
\and
Quanhui Zhu\thanks{
Department of Mathematics, National University of Singapore, Singapore
(\email{matcz@nus.edu.sg}, \email{zhuqh@nus.edu.sg}).
}
}

\usepackage{amsopn}

\ifpdf
\hypersetup{
  pdftitle={A General First- and Second-Order Numerical Solver for NMQSD},
  pdfauthor={Z. Cai and Q. Zhu}
}
\fi

\begin{document}

\maketitle

\begin{abstract}
The numerical simulation of non-Markovian open quantum systems based on the
non-Markovian quantum state diffusion (NMQSD) equation is complicated by
functional derivatives with respect to the stochastic process.
A general numerical framework that directly treats these functional derivatives without relying on prescribed decompositions of the bath correlation function is
still lacking.
In this work, we derive an analytical solution of the linear NMQSD equation that
reveals three elementary structures of the non-Markovian stochastic dynamics:
stochastic propagation, functional-derivative insertion, and memory pairing.
Based on this structure, we construct a general auxiliary-state framework
for arbitrary bath correlation functions.
The framework separates the numerical construction into time discretization,
memory quadrature, and hierarchy truncation.
We then construct first- and second-order schemes and provide
diagrammatic transition rules for their explicit implementation.
Numerical results verify the expected temporal accuracy and demonstrate the
applicability of the proposed methods to different bath correlation functions
and multi-level quantum systems.
\end{abstract}



\section{Introduction}
\label{sec:intro}
Open quantum system dynamics provides a framework for describing
the interaction between a quantum system and its environment \cite{breuer_nonmarkoviandynamicsphysical_2007,weiss_quantumdissipativesystems_2012, lindblad_generatorsquantumdynamical_1976, gorini_completelypositivedynamical_1976}. 
In regimes involving strong system--bath coupling, low temperatures, 
or structured environments, non-Markovian effects arise, 
with memory kernels and information backflow leading to intrinsically time-nonlocal dynamics \cite{leggett_dynamicsdissipativetwostate_1987, breuer_measuredegreenonmarkovian_2009, rivas_quantumnonmarkovianitycharacterization_2014, devega_dynamicsnonmarkovianopen_2017}.

Numerically exact approaches for non-Markovian dynamics include hierarchical density-matrix methods and path-integral formulations, such as 
hierarchical equations of motion,
quasi-adiabatic path integral methods, and tensor-network-based schemes \cite{tanimura_numericallyexactapproach_2020, makri_numericalpathintegral_1995, strathearn_efficientnonmarkovianquantum_2018, cerrillo_nonmarkoviandynamicalmaps_2014, pollock_operationalmarkovcondition_2018}. 
Although highly accurate, these methods typically suffer 
from rapidly increasing computational costs with memory time, 
bath complexity, or system dimension.
Recent developments aim to mitigate these limitations through more efficient treatments of influence functionals and improved discretization strategies \cite{cai_fastalgorithmsbath_2022, cai_inchwormmontecarlo_2020, cai_secondorderdiscretizationdyson_2025}.

An appealing alternative is provided by stochastic Schrödinger equations, 
which represent the reduced density matrix as an ensemble average over
stochastic pure-state trajectories.
Among such wave function approaches, the non-Markovian quantum state diffusion
(NMQSD) formalism provides a stochastic description of open quantum systems \cite{strunz_linearquantumstate_1996,diosi_nonmarkovianstochasticschrodinger_1997, diosi_nonmarkovianquantumstate_1998a, strunz_brownianmotionstochastic_2001, gaspard_nonmarkovianstochasticschrodinger_1999, strunz_opensystemdynamics_1999}.

The standard form of the linear NMQSD equation is given by \cite{diosi_nonmarkovianquantumstate_1998a}:
\begin{equation}
    \frac{\D}{\D t} \psi_t = -\mathrm{i} H \psi_t + L \psi_t z_t^* - L^\dagger \int_0^t \alpha(t,s) \frac{\delta\psi_t}{\delta z_s^*} \D s,
    \label{eq:NMQSD_linear}
\end{equation}
where $H$ is the system Hamiltonian, $L$ is the system--environment coupling operator, $\alpha(t,s)$ is the bath correlation function, and $z_t^*$ is a complex Gaussian stochastic process satisfying 
\begin{equation}
    \mathbb{E}[z_t]=0,\qquad
    \mathbb{E}[z_tz_s]=0,\qquad
    \mathbb{E}[z_tz_s^*]=\alpha(t,s).
\end{equation}
The wave function $\psi_t$ depends on the history of the stochastic process $Z = \{z_s: 0 \le s \le t\}$, and the ensemble average over trajectories recovers the reduced density matrix $\rho_t = \mathbb{E}[|\psi_t\rangle\langle\psi_t|]$.

A central challenge in~\eqref{eq:NMQSD_linear} lies in the functional derivative
$\delta\psi_t/\delta z_s^*$, which encodes the dependence of the wave function on
the past stochastic process.
Since this history is indexed by the continuous time variable $s\in[0,t]$, the
functional derivative represents a family of variations along the stochastic
trajectory. This dependence on a continuum of stochastic variables is
intrinsically infinite-dimensional, which is the main difficulty in constructing
general numerical solvers.

To address this challenge, several techniques have been developed. A prominent example is the hierarchy of pure states (HOPS), where the functional derivative is replaced by  a closed hierarchy of auxiliary wave functions when the bath correlation function can
be decomposed into sums of exponentials
\cite{suess_hierarchystochasticpure_2014, hartmann_exactopenquantum_2017, mascherpa_optimizedauxiliaryoscillators_2020}.
Related developments include tensor-network representations for many-body
systems, stochastic sampling methods, and variational techniques \cite{gao_nonmarkovianstochasticschrodinger_2022,stockburger_exact$mathitc$numberrepresentation_2002, yu_nonmarkovianquantumtrajectories_2004}. 
Machine learning approaches have also been explored for stochastic quantum dynamics \cite{carleo_solvingquantummanybody_2017, zhang_artificialintelligencebasedsurrogatesolution_2024, zhang_neuralnetworksolution_2025}.

However, a general numerical framework for NMQSD remains lacking.
Many existing auxiliary-state formulations rely on specific decompositions of the bath
correlation function and are therefore tied to particular classes of memory kernels.
For a general correlation function $\alpha(t,s)$, it remains unclear how to represent
the functional derivative $\delta\psi_t/\delta z_s^*$ in a closed and discretizable
hierarchy.
These limitations motivate a numerical framework that directly
discretizes the functional-derivative structure of NMQSD and supports systematically
accurate time discretizations for general bath correlations.

In this work, we derive a Dyson-type analytical solution of the linear NMQSD
equation. This solution reveals three elementary structures of the non-Markovian stochastic
dynamics: stochastic propagation, functional-derivative insertion, and memory
pairing.
Building on this structure, we develop a general auxiliary-state framework that
directly discretizes the functional derivative and yields first- and second-order
schemes.

The main contributions of this work are summarized as follows:
\begin{itemize}
\item \textbf{Analytical solution of NMQSD.}
We derive an analytical solution of the linear NMQSD equation and show that it
reconstructs the original stochastic equation.
The resulting representation reorganizes the non-Markovian dynamics in terms of
stochastic propagation, functional-derivative insertion, and bath-induced memory
pairing.

\item \textbf{General numerical framework.}
We formulate a general auxiliary-state framework for arbitrary bath correlation
functions by introducing auxiliary states generated by functional derivatives
with respect to the stochastic variables. For these auxiliary states, we derive
the evolution equations and insertion-time initial conditions.
The numerical construction is then organized by time discretization, memory
quadrature, and hierarchy truncation, leading to first- and second-order schemes.

\item \textbf{Diagrammatic interpretation.}
We introduce a diagrammatic representation that organizes the auxiliary-state
dynamics through propagation, insertion, and pairing transitions.
The diagrams express the first- and second-order schemes in explicit update
rules.

\end{itemize}

The paper is organized as follows.
Section~\ref{sec:analytical_solution} derives the analytical solution and
introduces its diagrammatic representation.
The auxiliary-state framework and the first- and second-order numerical
schemes are developed in Section~\ref{sec:methods}, while truncation and
computational complexity are discussed in
Section~\ref{sec:truncation_complexity}.
Numerical experiments are presented in Section~\ref{sec:results}.
Finally, Section~\ref{sec:conclusion} gives concluding remarks and discusses
future work.
The proofs of the main results and fully expanded examples of the hierarchical
updates are provided in the supplementary material.

\section{Analytical Solution}
 \label{sec:analytical_solution}

To simplify the subsequent analysis and numerical construction,
we pass to the interaction picture \cite{sakurai_modernquantummechanics_2020} via the transformation
\begin{equation}
    \tilde{\psi}_t := e^{iHt}\psi_t.
\end{equation}
Substituting this into~\eqref{eq:NMQSD_linear} eliminates
the drift term $-iH\psi_t$.
The coupling operator $L$ is correspondingly transformed into a time-dependent operator
$L(t) = e^{iHt} L e^{-iHt}$,
and the NMQSD equation in the interaction picture becomes
\begin{equation}
    \frac{\D}{\D t} \tilde{\psi}_t = z_t^* L(t) \tilde{\psi}_t - L^\dagger(t) \int_0^t \alpha(t,s) \frac{\delta\tilde{\psi}_t}{\delta z_s^*} \D s.
    \label{eq:NMQSD_interaction}
\end{equation}
For notational simplicity, we drop the tilde and write $\psi_t$ for the interaction-picture state in the following.
Unless otherwise specified, all quantities are understood in the interaction picture.

\subsection{Dyson-Type Expansion}
We first construct an exact analytical representation of the solution
to the linear NMQSD equation in the interaction picture:
\begin{equation}
    \frac{\D}{\D t} \psi_t = z_t^*L(t)\psi_t - L^\dagger(t) \int_0^t \alpha(t,s) \frac{\delta\psi_t}{\delta z_s^*} \D s.
    \label{eq:NMQSD_reduced}
\end{equation}
The solution can be written as a Dyson-type time-ordered expansion
\cite{breuer_theoryopenquantum_2002}. The following theorem gives its explicit form.

\begin{theorem}
\label{thm:solution}
The solution of~\eqref{eq:NMQSD_reduced} admits the expansion
\begin{equation}
    \psi_t
    =
    \sum_{n=0}^{\infty}
    \sum_{m=0}^{\infty}
    \psi_t^{(n,m)},
    \label{eq:solution}
\end{equation}
where \(\psi_t^{(0,0)}=\psi_0\) is the initial state, and for \((n,m)\neq(0,0)\),
\begin{equation}
\begin{aligned}
    \psi_t^{(n,m)}
    &=
    (-1)^m
    \int_{\Delta_n(t)}
    \D^n\boldsymbol{\tau}
    \int_{\Delta_{2m}(t)}
    f(\boldsymbol{\tau},\boldsymbol{s})
    \D^{2m}\boldsymbol{s},
    \\
    f(\boldsymbol{\tau},\boldsymbol{s})
    &=
    \left(\prod_{k=1}^n z_{\tau_k}^*\right)
    \sum_{P\in\mathcal{P}_{2m}}
    \mathcal{L}_P(s_1,\ldots,s_{2m})
    \\
    &\quad\times
    \mathcal{T}
    \left[
    \prod_{k=1}^n L(\tau_k)
    \prod_{j=1}^m
    L^\dagger(s_{P(2j)})L(s_{P(2j-1)})
    \right]\psi_0 .
\end{aligned}
\label{eq:psi_nm}
\end{equation}
The operator \(\mathcal T\) denotes time ordering over all system operators
according to their time arguments.
Moreover, the components satisfy
\begin{equation}
    \frac{\D}{\D t} \psi_t^{(n,m)}
    =
    z_t^* L(t)\psi_t^{(n-1,m)}
    -
    L^\dagger(t)
    \int_0^t
    \alpha(t,s)
    \frac{\delta \psi_t^{(n+1,m-1)}}{\delta z_s^*}
    \D s ,
    \label{eq:component_evolution}
\end{equation}
with the convention that \(\psi_t^{(n,m)}=0\) whenever \(n<0\) or \(m<0\).
Consequently, summing~\eqref{eq:component_evolution} over all
\((n,m)\) reconstructs~\eqref{eq:NMQSD_reduced}.
\end{theorem}
The simplex is defined by
\(
    \Delta_k(t)
    =
    \{(r_1,\ldots,r_k):0<r_1<\cdots<r_k<t\}.
\)
In the expansion, \(\Delta_n(t)\) and \(\Delta_{2m}(t)\) correspond to the
variables \(\boldsymbol{\tau}\) and \(\boldsymbol{s}\), respectively, with
\[
    \D^n\boldsymbol{\tau}
    =
    \D\tau_1\cdots\D\tau_n,
    \qquad
    \D^{2m}\boldsymbol{s}
    =
    \D s_1\cdots\D s_{2m}.
\]
The collection \(\mathcal P_{2m}\) denotes the set of all pairings of
\(\{1,\ldots,2m\}\). For \(P\in\mathcal P_{2m}\), we label its \(m\)
pairs as
\(
    (P(2j-1),P(2j)),\  j=1,\ldots,m,
\)
with \(P(2j-1)<P(2j)\), so that \(s_{P(2j-1)}\) is earlier than
\(s_{P(2j)}\) on the simplex \(\Delta_{2m}(t)\).
The associated bath-correlation factor is
\[
    \mathcal L_P(s_1,\ldots,s_{2m})
    =
    \prod_{j=1}^m
    \alpha(s_{P(2j)},s_{P(2j-1)}).
\]

The proof of Theorem~\ref{thm:solution} is obtained by taking the time derivative of
the component \(\psi_t^{(n,m)}\).
The boundary contribution from the \(\boldsymbol{\tau}\)-simplex generates
the stochastic term \(z_t^*L(t)\), while the boundary contribution from the
\(\boldsymbol{s}\)-simplex generates the memory term. 
The full derivation is provided in \ref{app:proof_solution} of the supplementary material.

Using the simplex structure and the time-ordering operator 
$\mathcal{T}$,
the stochastic part of the solution~\eqref{eq:solution}
can be resummed as 
\begin{equation}
\begin{aligned}
\sum_{n=0}^\infty\int_{\Delta_n(t)}\mathcal{T}
\left(
\prod_{k=1}^nz_{\tau_k}^*L(\tau_k)
\right)
\D^n\boldsymbol{\tau}&=
\sum_{n=0}^\infty
\frac{1}{n!}
\int_0^t \cdots \int_0^t
\mathcal{T}
\left(
\prod_{k=1}^n
z_{\tau_k}^*L(\tau_k)
\right)
\D^n\boldsymbol{\tau}\\
&=\mathcal{T}\exp\left(\int_0^t z_\tau^*L(\tau)\,\D\tau
\right).
\end{aligned}
\end{equation}
Then the solution can be rewritten in the compact form
\begin{equation}
    \psi_t=\mathcal{T}\left[\exp\left(\int_0^t z_s^* L(s)\,\D s\right)\mathcal{M}(t)\right]\psi_0,
\label{eq:solution_compact}
\end{equation}
where $\mathcal{M}(t)$ is a shorthand for the memory contribution
under the time-ordering operation:
\begin{equation}
\begin{aligned}
\mathcal{M}(t):=\sum_{m=0}^\infty (-1)^m\int_{\Delta_{2m}(t)}
\sum_{P \in \mathcal{P}_{2m}}\mathcal{L}_P(s_1,\dots,s_{2m})
\prod_{j=1}^m\left[L^\dagger(s_{P(2j)})L(s_{P(2j-1)})\right]\D^{2m}\boldsymbol{s}.
\end{aligned}
\end{equation}
Taking the functional derivative of the compact form~\eqref{eq:solution_compact} gives
\begin{equation}
\begin{aligned}
\frac{\delta}{\delta z_s^*}\psi_t
&=\frac{\delta}{\delta z_s^*}\mathcal{T}
\left[\exp\left(\int_0^t z_\tau^*L(\tau)\,\D\tau\right)\mathcal{M}(t)
\right]\psi_0
\\
&=\mathcal{T}
\left[L(s)\exp\left(\int_0^t z_\tau^*L(\tau)\,\D\tau\right)\mathcal{M}(t)
\right]\psi_0,
\qquad 0\le s\le t .
\end{aligned}
\label{eq:variational_insertion}
\end{equation}
Since \(\mathcal M(t)\) contains no explicit dependence on the stochastic process, the
functional derivative acts only on the time-ordered exponential. Thus, it corresponds to inserting an operator \(L(s)\) into the time-ordered product.

The compact representation \eqref{eq:solution_compact}, together with
the insertion rule \eqref{eq:variational_insertion}, identifies three
elementary structures:

\begin{itemize}

    \item \textbf{Stochastic propagation:}
    The stochastic term generates the time-ordered exponential
    \(
    \mathcal{T}
    \exp\left(
    \int_0^t z_s^*L(s)\,\D s
    \right),
    \)
    which represents the stochastic propagation of the wave function.

    \item \textbf{Operator insertion:}
    The functional derivative \(\delta/\delta z_s^*\) inserts an operator
    \(L(s)\) into the time-ordered product.

    \item \textbf{Pairing:}
    The bath correlation kernel
    \(\alpha(s_{P(2j)},s_{P(2j-1)})\) contracts an inserted operator
    \(L(s_{P(2j-1)})\) with the corresponding operator
    \(L^\dagger(s_{P(2j)})\), producing the nonlocal memory contribution
    encoded in \(\mathcal M(t)\).

\end{itemize}

Together, these three structures provide a compact interpretation of how the stochastic driving, functional differentiation, and bath memory are organized in the analytical solution.

\subsection{Diagrammatic Interpretation}
Diagrammatic notation is widely used to organize terms in perturbative expansions for open quantum systems, including Keldysh-contour formulations and Dyson-series-based diagrammatic Monte Carlo methods \cite{doi:10.1137/22M1499297,werner_diagrammaticmontecarlo_2009}. In the same spirit, we use diagrams as a graphical notation for the propagation, insertion, and pairing identified above. The corresponding graphical elements are defined as follows.
\begin{itemize}
    \item \textbf{Propagation:}
    \adjustbox{valign=c}{\includegraphics[scale=0.6]{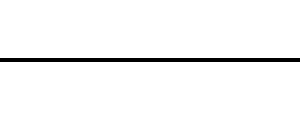}}

    The horizontal line represents stochastic propagation along the time axis.

    \item \textbf{Insertion:}
    \adjustbox{valign=c}{\includegraphics[scale=0.6]{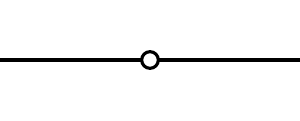}}

    An operator insertion at time \(s\) is represented by a circle $\circ$,
    corresponding to the operator \(L(s)\).

    \item \textbf{Pairing:}
    \adjustbox{valign=c}{\includegraphics[scale=0.6]{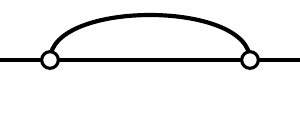}}

    A memory pairing is represented by an arch connecting the current point with a previous insertion circle.
It carries the bath weight \(\alpha(t,s)\) and the current operator
\(L^\dagger(t)\), while the paired \(L(s)\) comes from the previous insertion circle.
\end{itemize}

Using these graphical elements, the analytical solution~\eqref{eq:solution_compact} can be represented
as a sum over diagrams organized by the number of memory pairings \(m\):
\begin{equation}
\begin{aligned}
    (m=0)\quad & \psi_t(Z)
    \;\;=\;\;
    \adjustbox{valign=c}{\includegraphics[scale=0.6]{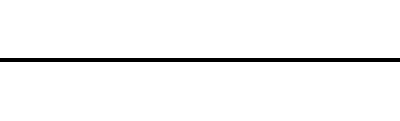}}
    \\
    (m=1)\quad & \;\;+\;\;
    \adjustbox{valign=c}{\includegraphics[scale=0.6]{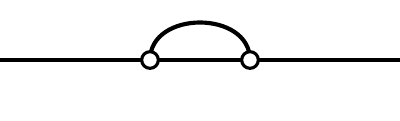}}
    \\
    (m=2)\quad & \;\;+\;\;
    \adjustbox{valign=c}{\includegraphics[scale=0.6]{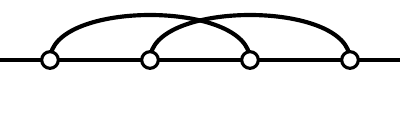}}
    \;\;+\;\;
    \adjustbox{valign=c}{\includegraphics[scale=0.6]{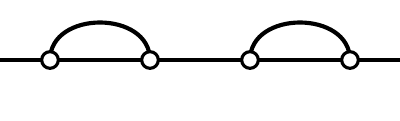}}
    \\
    & \;\;+\;\;
    \adjustbox{valign=c}{\includegraphics[scale=0.6]{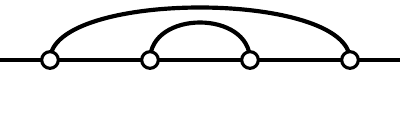}}
    \\
    (m\ge 3)\quad & \;\;+\;\;\cdots
\end{aligned}
\end{equation}

This diagrammatic expansion gives a compact visual organization of the
analytical solution. In particular, it makes explicit how stochastic
propagation, operator insertions, and memory pairings are combined
within the non-Markovian dynamics. 
Their correspondence with the NMQSD
equation, the analytical representation, and the diagrammatic notation is
summarized in Table~\ref{tab:structure_correspondence}.
These graphical elements will
serve as the basic building blocks for the hierarchical numerical schemes
developed in the next section.

\begin{table}[htbp]
\centering
\caption{
Correspondence among the NMQSD equation, the analytical representation,
and the diagrammatic notation.
}
\label{tab:structure_correspondence}

{
\setlength{\tabcolsep}{4pt}
\renewcommand{\arraystretch}{1.25}
\begin{tabular*}{\textwidth}{
@{\extracolsep{\fill}}
l l c c c
@{}
}
\toprule
\multicolumn{2}{r}{\textbf{Operation}}
&
\textbf{Propagation}
&
\textbf{Pairing}
&
\textbf{Insertion}
\\
\midrule

\multirow{2}{*}[-1ex]{\textbf{Equation}}
&
form
&
\multicolumn{3}{c}{
\(\displaystyle
\frac{\D}{\D t}\psi_t =z_t^*L(t)\psi_t-L^\dagger(t)\int_0^t\alpha(t,s)\frac{\delta\psi_t}{\delta z_s^*}
\,\D s\)
}\\
&term
&\(\displaystyle z_t^*L(t)\psi_t\)
&\(\displaystyle-L^\dagger(t)\int_0^t\alpha(t,s)
\frac{\delta\psi_t}{\delta z_s^*}\,\D s\)
&
\(\displaystyle\frac{\delta\psi_t}{\delta z_s^*}\)
\\ \midrule

\multirow{2}{*}[-1ex]{\textbf{Solution}}
&
form
&
\multicolumn{3}{c}{
\(\displaystyle
\psi_t
=
\mathcal T\left[
\exp\left(
\int_0^t z_s^*L(s)\,\D s
\right)
\mathcal M(t)
\right]\psi_0
\)
}
\\
&
term
&
\(\displaystyle
\mathcal T\exp\!\left(
\int_0^t z_s^*L(s)\,\D s
\right)
\)
&
\(\displaystyle
\mathcal M(t)
\)
&
\(\displaystyle
\mathcal T[L(s)\cdots]
\)
\\
\midrule

\textbf{Diagram}
&
&
\(\displaystyle
\adjustbox{valign=c}{
\includegraphics[scale=0.45]{figure/diagram/def/nmqsd_def-1.pdf}}
\)
&
\(\displaystyle
\adjustbox{valign=c}{
\includegraphics[scale=0.45]{figure/diagram/def/nmqsd_def-3.pdf}}
\)
&
\(\displaystyle
\adjustbox{valign=c}{
\includegraphics[scale=0.45]{figure/diagram/def/nmqsd_def-2.pdf}}
\)
\\

\bottomrule
\end{tabular*}
\vspace{0.5ex}
\begin{minipage}{0.96\textwidth}
\footnotesize
\textit{Remark.}
A single pair in \(\mathcal M(t)\) contributes
\(-\alpha(t,s)L^\dagger(t)L(s)\), where \(L^\dagger(t)\) acts at the current time \(t\), 
and the paired \(L(s)\) is generated by an insertion at the previous time \(s\).
\end{minipage}
}
\end{table}

\section{Numerical Methods}
\label{sec:methods}
We now turn the analytical representation developed above into numerical schemes. For this purpose, it is convenient to split the
right-hand side of~\eqref{eq:NMQSD_reduced} into the stochastic propagation
part and the memory part. Define
\begin{equation}
    A(s)=z_s^*L(s),
    \qquad
    B(s)=-L^\dagger(s)\int_0^s\alpha(s,u)
    \frac{\delta}{\delta z_u^*}\D u .
\end{equation}
Then equation~\eqref{eq:NMQSD_reduced} can be written as
\begin{equation}
    \frac{\D}{\D t}\psi_t=
    \bigl(A(t)+B(t)\bigr)\psi_t .
    \label{eq:equation_numerical}
\end{equation}
The following subsections introduce the auxiliary states and then construct the
first- and second-order schemes.

\subsection{Auxiliary States}

The memory operator \(B(t)\) involves functional derivatives of the wave
function with respect to the past stochastic variables. To obtain a closed
discrete representation of this term, we introduce auxiliary states
\begin{equation}
    \psi_t^\sigma:=\frac{\delta^k \psi_t}{\delta z_{s_1}^* \cdots \delta z_{s_k}^*},
    \qquad
    \sigma=(s_1,\ldots,s_k).
    \label{eq:auxiliary_state}
\end{equation}
Here \(\sigma\) is an ordered index satisfying
\(0\le s_1\le\cdots\le s_k\le t\).
The physical state corresponds to the empty index,
\(\psi_t^\emptyset=\psi_t\).

For an ordered index \(\sigma\), the corresponding auxiliary state $\psi_t^\sigma$ 
satisfies the same evolution equation and is initialized at \(t=s_k\):
\begin{equation}
\begin{aligned}
    \frac{\D}{\D t}\psi_t^\sigma
    &=\bigl(A(t)+B(t)\bigr)\psi_t^\sigma,
    && t>s_k, \\
    \psi_{s_k}^\sigma
    &=L(s_k)\,\psi_{s_k}^{\sigma\setminus(s_k)},
    && t=s_k .
\end{aligned}
\label{eq:aux_states}
\end{equation}
The initial condition in~\eqref{eq:aux_states} is precisely the
insertion of \(L(s_k)\) at time \(s_k\). Indeed, if
\(\sigma_-=(s_1,\ldots,s_{k-1})\), then adding the new time \(s_k\) gives
\begin{equation}
    \psi_{s_k}^{\sigma_-\cup(s_k)}
    =L(s_k)\psi_{s_k}^{\sigma_-},
    \qquad t=s_k .
    \label{eq:aux_insertion_initial}
\end{equation}
Here \(\sigma\setminus(s_k)\) denotes the ordered index obtained by removing
\(s_k\) from \(\sigma\), while \(\sigma_-\cup(s_k)\) denotes the ordered index obtained by adding \(s_k\) to \(\sigma_-\).

Thus, the auxiliary states are generated causally: each state is initialized at its last index time by inserting \(L(s_k)\), and is then evolved by the same equation. The numerical construction is determined by three choices: (i) the time discretization of this evolution equation, (ii) the quadrature rule for the memory operator, and (iii) the selection of the retained auxiliary states. The first- and second-order schemes are presented next under the following assumptions.

\begin{assumption}
\label{assumption}
For any fixed final time \(T>0\), the following conditions hold.

\begin{enumerate}[label=(\arabic*)]
    \item \textbf{Bounded coefficients.}
    All bounds involving the stochastic process are understood pathwise.
    There exist constants \(C_A,C_L>0\) such that
    \[
        \sup_{0\le t\le T}\|A(t)\|\le C_A,
        \qquad
        \sup_{0\le t\le T}\|L(t)\|\le C_L .
    \]
\item \textbf{Bounded memory kernel.}
There exists a constant \(C_\alpha>0\) such that
\[
    |\alpha(t,s)|\le C_\alpha,
    \qquad (t,s)\in[0,T]^2 .
\]
The memory kernel is then uniformly integrable on finite time
intervals:
\[
    \sup_{0\le t\le T}
    \int_0^t |\alpha(t,s)|\,\D s
    \le C_{\alpha,T}.
\]

\item \textbf{Level-wise boundedness of the auxiliary hierarchy.}
For the continuous auxiliary states defined in \eqref{eq:auxiliary_state},
there exist constants \(C_\psi,M>0\) such that
\[
    \sup_{0\le t\le T}
    \sup_{|\sigma|=k}
    \|\psi_t^\sigma\|
    \le C_\psi M^k,
    \qquad k\ge 0 .
\]
\end{enumerate}
\end{assumption}
The boundedness of \(L(t)\) and \(\alpha(t,s)\) is a standard finite-time
assumption. Since \(A(t)=z_t^*L(t)\) contains the stochastic process, its
boundedness is understood pathwise, namely for almost every realization.
The level-wise boundedness condition controls the growth of auxiliary states
with respect to the hierarchy level.

\subsection{First-Order Scheme}
We construct the first-order auxiliary-state hierarchy using forward Euler time
discretization and a left-endpoint quadrature for the memory integral. In this first-order construction, the auxiliary states are indexed by distinct time labels, while the repeated-label contributions are dropped.

\subsubsection{Construction of the First-Order Scheme}
We first construct a first-order discretization of~\eqref{eq:aux_states}. Let the time interval be discretized by a uniform grid \(t_n=n\tau\), where \(\tau\) is the time step. For notational simplicity, we write
\begin{equation}
    z_n^* = z_{t_n}^*, \quad
    L_n = L(t_n), \quad
    \alpha_{n,m} = \alpha(t_n,t_m), \quad
    A_n = A(t_n), \quad
    \psi_n^\sigma = \psi_{t_n}^\sigma .
\end{equation}
Let \(B(t_n)\) denote the continuous memory operator,
and \(B_n\) denote its discrete approximation.
Applying the forward Euler discretization to~\eqref{eq:aux_states} gives
\begin{equation}
    \psi_{n+1}^\sigma = \psi_n^\sigma +
    \tau\left(A_n\psi_n^\sigma+B(t_n)\psi_n^\sigma\right)
    +\mathcal O(\tau^2).
    \label{eq:first_order_euler_continuous_B}
\end{equation}
The memory term is approximated by the left-endpoint quadrature
\begin{equation}
\begin{aligned}
    B(t_n)\psi_n^\sigma
    &= -L_n^\dagger\int_0^{t_n}\alpha(t_n,s)\frac{\delta\psi_n^\sigma}{\delta z_s^*}\D s\\
    &=-\tau L_n^\dagger\sum_{j=0}^{n-1}\alpha_{n,j}\frac{\delta\psi_n^\sigma}{\delta z_{t_j}^*}
    +\mathcal O(\tau)\\
    &=-\tau L_n^\dagger\sum_{j=0}^{n-1}\alpha_{n,j}\psi_n^{\sigma\cup(t_j)}
    +\mathcal O(\tau)\\
\end{aligned}
\label{eq:first_order_memory_continuous}
\end{equation}
We then set
\begin{equation}
    B_n\psi_n^\sigma:=-\tau L_n^\dagger\sum_{j=0}^{n-1}\alpha_{n,j}\psi_n^{\sigma\cup(t_j)}.
    \label{eq:first_order_memory_discrete}
\end{equation}
Dropping the remainder gives the full first-order update
\begin{equation}
\begin{aligned}
    \psi_{n+1}^\sigma &=(I+\tau A_n)\psi_n^\sigma +\tau B_n\psi_n^\sigma  \\
    &=(I+\tau A_n)\psi_n^\sigma -
    \tau^2 L_n^\dagger\sum_{j=0}^{n-1}\alpha_{n,j}\psi_n^{\sigma\cup(t_j)}.
\end{aligned}
\label{eq:first_order_update}
\end{equation}
In the first-order update from \(t_n\) to \(t_{n+1}\), the memory quadrature
uses the past grid points
\[
    \mathcal I_n^{[1]}:=\{t_0,t_1,\ldots,t_{n-1}\}.
\]
The retained first-order configuration set is the distinct-index set
\begin{equation}
    \mathcal C_n^{[1]}:=
    \left\{
    \sigma=(s_1,\ldots,s_k):
    s_j\in\mathcal I_n^{[1]},
    \ s_i\ne s_j \ (i\ne j)
    \right\}.
    \label{eq:first_order_configuration_set}
\end{equation}
Since $\sigma$ is ordered, every $\sigma\in \mathcal C_n^{[1]}$ satisfies $s_1 < \cdots < s_k$.
At time \(t_n\), the implemented first-order hierarchy only retains
auxiliary states \(\psi_n^\sigma\) with
\(\sigma\in\mathcal C_n^{[1]}\).

In the distinct-index implementation, the memory operator is restricted
to those terms whose added label remains admissible.  We therefore define
the first-order retained memory operator by
\begin{equation}
    B_n^{[1]}\psi_n^\sigma:=
    -\tau L_n^\dagger
    \sum_{\substack{0\le j\le n-1\\ t_j\notin\sigma}}
    \alpha_{n,j}
    \psi_n^{\sigma\cup(t_j)},
    \qquad
    \sigma\in\mathcal C_n^{[1]} .
    \label{eq:first_order_retained_Bn}
\end{equation}
Then \(\sigma\cup(t_j)\) still belongs to \(\mathcal C_n^{[1]}\), and
the implemented first-order scheme is
\begin{equation}
\begin{aligned}
    \psi_{n+1}^\sigma
    &=(I+\tau A_n)\psi_n^\sigma+\tau B_n^{[1]}\psi_n^\sigma  \\
    &=(I+\tau A_n)\psi_n^\sigma
    -\tau^2 L_n^\dagger
    \sum_{\substack{0\le j\le n-1\\ t_j\notin\sigma}}
    \alpha_{n,j}
    \psi_n^{\sigma\cup(t_j)},
    \qquad
    \sigma\in\mathcal C_n^{[1]} .
\end{aligned}
\label{eq:first_order_distinct_update}
\end{equation}
The new label \(t_n\) is added to the configuration set by the discrete
initial condition:
\begin{equation}
    \psi_{n+1}^{\sigma\cup(t_n)}
    =
    L_n\psi_n^\sigma,
    \qquad
    \sigma\in\mathcal C_n^{[1]},
    \quad
    \sigma\cup(t_n)\in\mathcal C_{n+1}^{[1]} .
    \label{eq:first_order_initial_condition_discrete}
\end{equation}
Together, \eqref{eq:first_order_distinct_update} and
\eqref{eq:first_order_initial_condition_discrete} define the first-order numerical scheme for \eqref{eq:aux_states}.

For comparison, let \(\phi_n^\sigma\) denote the solution of the full first-order Euler
scheme with the same initialization rule.  It satisfies
\begin{equation}
    \phi_{n+1}^\sigma
    =
    (I + \tau A_n)\phi_n^\sigma
    +
    \tau B_n\phi_n^\sigma.
    \label{eq:first_order_full_hierarchy}
\end{equation}
The next theorem shows that, at the level of the physical state, the error introduced by restricting the memory operator to the distinct-index configurations is \(O(\tau)\).

\begin{theorem}[Consistency of the retained first-order memory operator]
\label{thm:error_first_order}
Under Assumption~\ref{assumption}, let \(\phi_n^\sigma\) be the solution of
the full first-order Euler hierarchy \eqref{eq:first_order_full_hierarchy},
where all configurations are retained, and let \(\psi_n^\sigma\) be the
solution of the implemented first-order scheme
\eqref{eq:first_order_distinct_update}.  Both schemes use the same
discrete initialization rule
\eqref{eq:first_order_initial_condition_discrete}.  Assume that
\[
    \psi_0^\emptyset=\phi_0^\emptyset = \psi_0.
\]
Then there exists a constant \(C_T>0\) such that
\begin{equation}
    \max_{0\le n\le N}
    \|\psi_n^\emptyset-\phi_n^\emptyset\|
    \le C_T\tau .
    \label{eq:physical_state_repeated_error}
\end{equation}
\end{theorem}

The proof is provided in \ref{app:first_order_consistency} of the supplementary material. Consequently, if the full hierarchy \eqref{eq:first_order_full_hierarchy} is first-order accurate for the physical state, then the implemented scheme \eqref{eq:first_order_distinct_update} preserves the first-order accuracy.

\subsubsection{Diagrammatic Implementation}

We now describe the diagrammatic representation of the first-order scheme.
For an auxiliary state \(\psi_n^\sigma\) with
\[
    \sigma=(s_1,\ldots,s_k)\in\mathcal C_n^{[1]},
    \qquad
    0\le s_1<\cdots<s_k\le t_{n-1},
\]
the distinct time labels in \(\sigma\) are represented diagrammatically by
\begin{equation}
\psi_t^\sigma = 
\begin{tikzpicture}[baseline=(img.base)]
    \node[inner sep=0pt] (img) {\adjustbox{valign=c}{\includegraphics[scale=0.6]{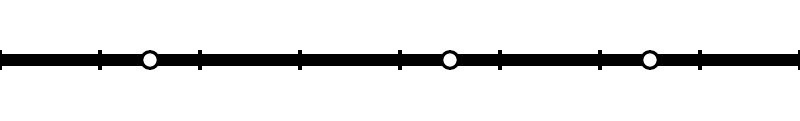}}};
    
    \node[below=0pt, font=\small] at (img.south west) {$0$};
    \node[below=0pt, font=\small] at (img.south east) {$t$};
    
    \path (img.south west) -- (img.south east) 
        coordinate[pos=0.1875] (s1)
        coordinate[pos=0.5625] (s2)
        coordinate[pos=0.7] (sc)
        coordinate[pos=0.8125] (s3);
        
    \node[below=2pt, font=\small, text=black] at (s1) {$s_1$};
    \node[below=2pt, font=\small, text=black] at (s2) {$s_2$};
    \node[below=2pt, font=\small, text=black] at (sc) {$\cdots$};
    \node[below=2pt, font=\small, text=black] at (s3) {$s_k$};
\end{tikzpicture}
\end{equation}
The bold line represents the current auxiliary state, and the circles
\(\circ\) record the unpaired time labels in
\(\sigma\). Their positions are schematic and
indicate only the ordering of the labels. In the first-order scheme, the
corresponding operations are evaluated at the left endpoint.
Since repeated time labels are not included, each auxiliary configuration can be encoded by a binary vector over
the active time grid.

The first-order scheme \eqref{eq:first_order_distinct_update} is written in a
pull-back form: for a fixed target state \(\psi_{n+1}^\sigma\), the memory term
collects contributions from higher-level states
\(\psi_n^{\sigma\cup(t_j)}\) with \(t_j\notin\sigma\).
For implementation, it is more convenient to use the equivalent push-forward
form. Namely, each active auxiliary state \(\psi_n^\sigma\), with
\(\sigma\in\mathcal C_n^{[1]}\), distributes its contributions to the
corresponding target states at time \(t_{n+1}\).

For a given auxiliary state \(\psi_n^\sigma\) at time \(t_n\), the first-order
scheme generates three elementary diagrammatic transitions: propagation,
insertion, and pairing. Here ``\(\leftarrow\)'' denotes an accumulated
contribution to an existing target state, rather than a complete state
assignment.  The diagrams below display these elementary contributions.

\adjustboxset{valign=m, height=2em, raise=-0.2ex}
\begin{itemize}    
    \item \textbf{Propagation:} Contribution to $\psi_{n+1}^\sigma$.
    \begin{flalign}&
\begin{alignedat}{2}
    &\text{Rule}\qquad
    & \psi^{\sigma}_{n+1}
    &\;\leftarrow\;
    \left( I + \tau A_n \right)\psi^{\sigma}_n
    \\
    &\text{Examples}&&
    \\[-1ex]
    &\quad\sigma=\emptyset:\qquad
    & \adjustbox{valign=c}{\includegraphics[scale=0.6]{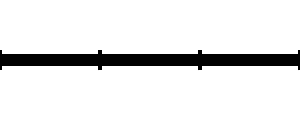}}
    &\;=\; \adjustbox{valign=c}{\includegraphics[scale=0.6]{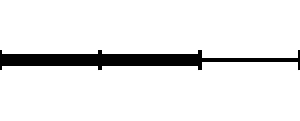}}
    + \cdots
    \\[-3ex]
    &\quad\sigma=(t_0):\qquad
    & \adjustbox{valign=c}{\includegraphics[scale=0.6]{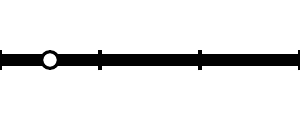}}
    &\;=\; \adjustbox{valign=c}{\includegraphics[scale=0.6]{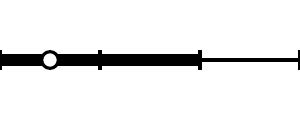}}
    + \cdots 
\end{alignedat}&&
\end{flalign}
This transition represents the time evolution generated by \(A_n=z_n^*L_n\).

    \item \textbf{Insertion:} Contribution to $\psi_{n+1}^{(\sigma,t_n)}$.
\begin{flalign}&
\begin{alignedat}{2}
    &\text{Rule}\qquad
    & \psi^{(\sigma,t_n)}_{n+1}
    &\;=\; L_n \psi^{\sigma}_n
    \\
    &\text{Examples}&&
    \\[-1ex]
    &\quad \sigma=\emptyset:\qquad
    & \adjustbox{valign=c}{\includegraphics[scale=0.6]{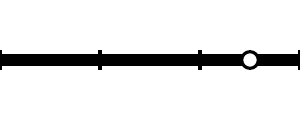}}
    &\;=\; \adjustbox{valign=c}{\includegraphics[scale=0.6]{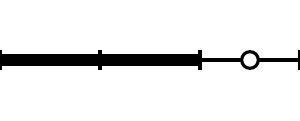}}
    \\[-3ex]
    &\quad \sigma=(t_0):\qquad
    & \adjustbox{valign=c}{\includegraphics[scale=0.6]{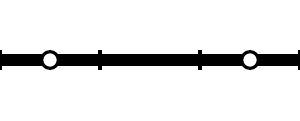}}
    &\;=\; \adjustbox{valign=c}{\includegraphics[scale=0.6]{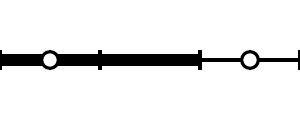}}
\end{alignedat}&&
\end{flalign}

It initializes the auxiliary state associated with the newly inserted
time label \(t_n\).
    \item \textbf{Pairing:} Contribution to $\psi_{n+1}^{\sigma \setminus (s)}$ for each $s\in\sigma$.
\begin{flalign}
&\begin{alignedat}{2}
    &\text{Rule}\qquad
    & \psi^{\sigma\setminus(s)}_{n+1}
    &\;\leftarrow\;
    - \tau^2 \,\alpha(t_n,s) L^\dagger_n \psi^{\sigma}_n
    \\
    &\text{Examples}
    &&
    \\[-1ex]
    &\quad \sigma=(t_0),\ s=t_0:\qquad
    & \adjustbox{valign=c}{\includegraphics[scale=0.6]{figure/diagram/first_order/nmqsd_first_order-1.pdf}}
    &\;=\; \adjustbox{valign=c}{\includegraphics[scale=0.6]{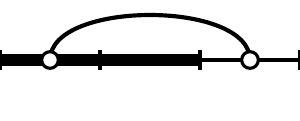}}
    + \cdots
    \\[-3ex]
    &\quad \sigma=(t_0,t_1),\ s=t_1:\qquad
    & \adjustbox{valign=c}{\includegraphics[scale=0.6]{figure/diagram/first_order/nmqsd_first_order-4.pdf}}
    &\;=\; \adjustbox{valign=c}{\includegraphics[scale=0.6]{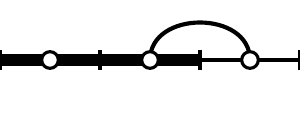}}
    + \cdots
\end{alignedat}
&&
\end{flalign}
This transition pairs the current point at \(t_n\) with the circle labeled
by \(s\).  After the pairing, this circle is removed, and the contribution
is accumulated into the lower-level state
\(\psi_{n+1}^{\sigma\setminus(s)}\).
\end{itemize}

These transitions completely specify the diagrammatic first-order update. The
corresponding algorithm is given in Algorithm~\ref{alg:first_order_nmqsd}, 
and a fully expanded one-step example is provided in
the supplementary material.
\begin{algorithm}[htbp]
\caption{First-Order Hierarchical Evolution Step}
\label{alg:first_order_nmqsd}
\begin{algorithmic}[1]
\Require Auxiliary states $\{\psi^{\sigma}_n\}_{\sigma \in \mathcal{C}_n^{[1]}}$ at time $t_n$, stochastic process $z_n$, time step $\tau$, bath correlation function $\alpha(t, s)$.
\Ensure Auxiliary states $\{\psi^{\sigma}_{n+1}\}_{\sigma \in \mathcal{C}_{n+1}^{[1]}}$ at time $t_{n+1}$.

\State Initialize target states: $\psi^{\sigma}_{n+1} = 0$ for all configurations $\sigma \in \mathcal{C}_{n+1}^{[1]}$.

\For{each active configuration $\sigma \in \mathcal{C}_n^{[1]}$}

    \State \textbf{Propagation: } $\psi^{\sigma}_{n+1} \mathrel{+\!=} \left( I + z^*_n L_n \tau \right) \psi^{\sigma}_n$. 
    
    \State \textbf{Insertion: } $\psi^{(\sigma,t_n)}_{n+1} \mathrel{+\!=} L_n \psi^{\sigma}_n$.
    \For{each $s \in \sigma$}
         \State \textbf{Pairing: }$\psi^{\sigma\setminus(s)}_{n+1} \mathrel{+\!=} - \tau^2 \, \alpha(t_n,s) L^\dagger_n \psi^{\sigma}_n$.
    \EndFor
\EndFor

\State \Return $\{\psi^{\sigma}_{n+1}\}_{\sigma \in \mathcal{C}_{n+1}^{[1]}}$
\end{algorithmic}
\end{algorithm}

At the physical level, the first-order update can be viewed as the combined
effect of stochastic propagation and memory pairing:
\begin{equation}
\begin{aligned}
     &\quad \text{Propagation} &&\quad\text{Pairing} \\[1.5ex]
    \psi_{n+1} &= (I+\tau A_n)\psi_n &&    -\tau^2 L_n^\dagger
    \sum_{\substack{0\le j\le n-1}}
    \alpha_{n,j}\psi_n^{(t_j)} \\
    \adjustbox{valign=c}{\includegraphics[scale=0.6]{figure/diagram/first_order/nmqsd_first_order-1.pdf}} 
    &= \adjustbox{valign=c}{\includegraphics[scale=0.6]{figure/diagram/first_order/nmqsd_first_order-9.pdf}} 
    &&+ \adjustbox{valign=c}{\includegraphics[scale=0.6]{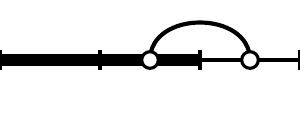}} 
    + \adjustbox{valign=c}{\includegraphics[scale=0.6]{figure/diagram/first_order/nmqsd_first_order-12.pdf}}    
\end{aligned}
\end{equation}

The insertion transition does not contribute directly to the physical state at
the same time step. Instead, it initializes higher-level
auxiliary states that contribute through later memory pairings. 
Thus, propagation, insertion, and pairing together form the complete
first-order scheme of the auxiliary states.

\subsection{Second-Order Scheme}
\subsubsection{Construction of the Second-Order Scheme}
Following the first-order construction, we obtain a second-order scheme by
combining a symmetric Strang splitting with a midpoint quadrature for the
memory integral. To close the resulting discrete hierarchy, we retain
distinct-label configurations together with configurations containing at
most one repeated time label.

Starting from the auxiliary-state equation \eqref{eq:aux_states}, we first apply
the symmetric Strang splitting in which the time-dependent operators are evaluated at \(t_{n+\frac12}\):
\begin{equation}
\psi_{n+1}^\sigma
=
\exp\left(
\frac{\tau}{2}A_{n+\frac12}
\right)
\exp\left(
\tau B(t_{n+\frac12})
\right)
\exp\left(
\frac{\tau}{2}A_{n+\frac12}
\right)
\psi_n^\sigma + \mathcal{O}(\tau^3).
\label{eq:second_order_strang_splitting}
\end{equation}
For later use, define the intermediate auxiliary state after the first
half-step stochastic propagation by
\begin{equation}
\hat{\psi}_{n+\frac12}^{\sigma}
:=
\exp\left(
\frac{\tau}{2}A_{n+\frac12}
\right)
\psi_n^\sigma .
\label{eq:second_order_midpoint_notation}
\end{equation}

The continuous memory operator acting on this intermediate state is
\begin{equation}
B(t_{n+\frac12}) \hat{\psi}_{n+\frac12}^{\sigma}
=
- L_{n+\frac12}^{\dagger}
\int_0^{t_{n+\frac12}}
\alpha(t_{n+\frac12},s)
\frac{\delta \hat{\psi}_{n+\frac12}^{\sigma}}{\delta z_s^*}
\,\D s .
\label{eq:second_order_continuous_midpoint_memory}
\end{equation}
We approximate the integral in \eqref{eq:second_order_continuous_midpoint_memory}
by the midpoint rule on \([0,t_n]\), together with a first-order approximation
on the half interval \([t_n,t_{n+\frac12}]\):
\begin{equation}
B_{n+\frac12}\hat{\psi}_{n+\frac12}^{\sigma}:=
- \tau L_{n+\frac12}^{\dagger}
\sum_{j=0}^{n-1}
\alpha_{n+\frac12,j+\frac12}
\hat{\psi}_{n+\frac12}^{\sigma\cup(t_{j+\frac12})}
-\frac{\tau}{2}
\alpha_{n+\frac12,n+\frac12}
L_{n+\frac12}^{\dagger}
\hat{\psi}_{n+\frac12}^{\sigma\cup(t_{n+\frac12})}.
\label{eq:second_order_discrete_memory}
\end{equation}
Since the quadrature error in the memory operator is \(O(\tau^2)\), we have
\begin{equation}
\tau B(t_{n+\frac12})\hat{\psi}_{n+\frac12}^{\sigma}
=\tau B_{n+\frac12}\hat{\psi}_{n+\frac12}^{\sigma}+\mathcal O(\tau^3).
\label{eq:second_order_memory_consistency}
\end{equation}

It remains to specify which auxiliary states appearing in
\(B_{n+\frac12}\hat{\psi}_{n+\frac12}^{\sigma}\) are retained.
At time \(t_n\), the available labels are chosen from
\[
    \mathcal I_n^{[2]}:=
    \{t_{\frac12},t_{\frac32},\ldots,t_{n-\frac12}\}.
\]
We first define the distinct configurations by
\[
    \mathcal C_n^{[2]}:=
    \left\{
    \sigma=(s_1,\ldots,s_k):
    s_j\in\mathcal I_n^{[2]},
    \ s_i\ne s_j \ (i\ne j)
    \right\}.
\]
Motivated by the second-order Dyson discretization for non-Markovian
density-matrix dynamics
\cite{cai_secondorderdiscretizationdyson_2025}, we further allow at most one
repeated time label in each configuration. The retained second-order configuration set is  then defined by
\begin{equation}
\begin{aligned}
    \mathcal D_n^{[2]}:=
    \mathcal C_n^{[2]}
    \cup
    \left\{
    \sigma\cup(s,s):
    \sigma\in\mathcal C_n^{[2]},
    \ s\in\mathcal I_n^{[2]},
    \ s\notin\sigma
    \right\}.
\end{aligned}
\label{eq:second_order_configuration_set}
\end{equation}
Here a repeated label \(s\) represents the repeated functional derivative \(\delta^2/\delta (z_s^*)^2\). Configurations with more than one repeated label, or with any label of multiplicity greater than two, are not retained.

During the step from \(t_n\) to \(t_{n+1}\), the insertion at the current
time \(t_{n+\frac12}\) enlarges the available label set from
\(\mathcal I_n^{[2]}\) to \(\mathcal I_{n+1}^{[2]}\). Therefore, after the
insertion step, the memory operator is restricted to configurations in
\(\mathcal D_{n+1}^{[2]}\). For \(\sigma\in\mathcal D_{n+1}^{[2]}\), define
\begin{equation}
\begin{aligned}
B_{n+\frac12}^{[2]}\hat{\psi}_{n+\frac12}^{\sigma}:=
&-\tau L_{n+\frac12}^{\dagger}
\sum_{\substack{0\le j\le n-1\\
\sigma\cup(t_{j+\frac12})\in\mathcal D_{n+1}^{[2]}}}
\alpha_{n+\frac12,j+\frac12}
\hat{\psi}_{n+\frac12}^{\sigma\cup(t_{j+\frac12})}
\\
&-\frac{\tau}{2}\alpha_{n+\frac12,n+\frac12}
L_{n+\frac12}^{\dagger}\hat{\psi}_{n+\frac12}^{\sigma\cup(t_{n+\frac12})},
\end{aligned}
\label{eq:second_order_retained_memory}
\end{equation}
where the last term is included only when
\(\sigma\cup(t_{n+\frac12})\in\mathcal D_{n+1}^{[2]}\).
This defines the retained memory operator. The auxiliary states involving
the current midpoint label \(t_{n+\frac12}\) are generated by the insertion
rule. For \(\sigma\in\mathcal D_n^{[2]}\),
\begin{equation}
    \hat{\psi}_{n+\frac12}^{\sigma\cup(t_{n+\frac12})}
    =L_{n+\frac12}\hat{\psi}_{n+\frac12}^{\sigma},
    \qquad
    \hat{\psi}_{n+\frac12}^{\sigma\cup(t_{n+\frac12},t_{n+\frac12})}
    =L_{n+\frac12}^{2}\hat{\psi}_{n+\frac12}^{\sigma}.
\label{eq:second_order_midpoint_insertions}
\end{equation}
These insertion rules are applied only to admissible configurations in
\(\mathcal D_{n+1}^{[2]}\).

We can now summarize the second-order scheme. Let
\(
    U_{n+\frac12}:=
    \exp\left(\frac{\tau}{2}A_{n+\frac12}\right).
\)
Starting from \(\{\psi_n^\sigma\}_{\sigma\in\mathcal D_n^{[2]}}\), the one-step update is
\begin{equation}
\begin{array}{lll}
\text{Right half-step:}
&
\hat{\psi}_{n+\frac12}^{\sigma}
=
U_{n+\frac12}\psi_n^\sigma,
&\sigma\in\mathcal D_n^{[2]},
\\
\text{Insertion:}
&
\text{apply the insertion rules in }
\eqref{eq:second_order_midpoint_insertions},
&\sigma\in\mathcal D_n^{[2]},
\\
\text{Memory pairing:}
&
\widetilde{\psi}_{n+\frac12}^{\sigma}
=
\left(I+\tau B_{n+\frac12}^{[2]}\right)
\hat{\psi}_{n+\frac12}^{\sigma},&
\sigma\in\mathcal D_{n+1}^{[2]},
\\
\text{Left half-step:}
&
\psi_{n+1}^{\sigma}
=
U_{n+\frac12}\widetilde{\psi}_{n+\frac12}^{\sigma},
&\sigma\in\mathcal D_{n+1}^{[2]} .
\end{array}
\label{eq:second_order_update_summary}
\end{equation}
Here \(B_{n+\frac12}^{[2]}\) is defined in
\eqref{eq:second_order_retained_memory}; its action on
\(\hat{\psi}_{n+\frac12}^{\sigma}\) involves the retained higher-level
auxiliary states \(\hat{\psi}_{n+\frac12}^{\sigma\cup(s)}\).

The Strang splitting and midpoint quadrature define a second-order
discretization at the level of the full discrete hierarchy. It remains to
show that replacing the full memory operator $B_{n+\frac12}$ by the retained operator
\(B_{n+\frac12}^{[2]}\) preserves this accuracy for the physical component. The following theorem provides this consistency result.

\begin{theorem}[Consistency of the retained second-order memory operator]
\label{thm:second_order_retained_memory_consistency}
Under Assumption~\ref{assumption}, let \(\phi_n^\sigma\) denote the solution of the full discrete second-order hierarchy
 obtained from \eqref{eq:second_order_update_summary} by replacing
\(B_{n+\frac12}^{[2]}\) with the full discrete memory operator
\(B_{n+\frac12}\).  Let \(\psi_n^\sigma\) denote the solution of the retained
second-order hierarchy defined by \eqref{eq:second_order_update_summary}.  The two hierarchies use the same propagation and insertion rules, and the
same initial physical state:
\[
    \psi_0^\emptyset=\phi_0^\emptyset=\psi_0 .
\]

Then there exists a constant \(C_T>0\) such that
\[
    \max_{0\le n\le N}
    \left\|
    \psi_n^\emptyset-\phi_n^\emptyset
    \right\|
    \le
    C_T\tau^2,
    \qquad N\tau\le T .
\]
\end{theorem}
The proof is given in \ref{app:second_order_consistency} of the supplementary material. Theorem~\ref{thm:second_order_retained_memory_consistency} shows that replacing the full memory operator \(B_{n+\frac12}\) by its retained counterpart \(B_{n+\frac12}^{[2]}\) introduces only an
\(\mathcal O(\tau^2)\) error in the physical component. Hence the retained
hierarchy preserves the second-order accuracy of the full discrete scheme at
the physical level.

\subsubsection{Diagrammatic Implementation}
The second-order scheme uses the same elementary diagrammatic operations as
the first-order scheme, but extends the notation by representing a repeated
time label with a double circle at the corresponding time point. Accordingly,
the auxiliary state \(\psi_t^\sigma\) is represented by
\begin{equation}
\psi_t^{\sigma} = 
\begin{tikzpicture}[baseline=(img.base)]
\node[inner sep=0pt] (img) {\adjustbox{valign=c}{\includegraphics[scale=0.8]{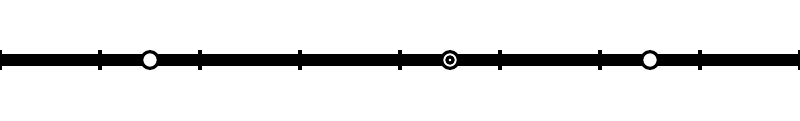}}};
    
\node[below=0pt, font=\small] at (img.south west) {$0$};
    \node[below=0pt, font=\small] at (img.south east) {$t$};
    
\path (img.south west) -- (img.south east) 
        coordinate[pos=0.1875] (s1)
        coordinate[pos=0.375] (sc1)
        coordinate[pos=0.5625] (s2)
        coordinate[pos=0.7] (sc)
        coordinate[pos=0.8125] (s3);
        
\node[below=2pt, font=\small, text=black] at (s1) {$s_1$};
    \node[below=2pt, font=\small, text=black] at (sc1) {$\cdots$};
    \node[below=2pt, font=\small, text=black] at (s2) {$s_k,s_{k+1}$};
    \node[below=2pt, font=\small, text=black] at (sc) {$\cdots$};
    \node[below=2pt, font=\small, text=black] at (s3) {$s_n$};
\end{tikzpicture}
\end{equation}
where the circles \(\circ\) denote the labels contained in \(\sigma\),
and the double circle \(\circledcirc\) marks the repeated label, if
present.

Following the Strang splitting in
\eqref{eq:second_order_strang_splitting}, each diagrammatic transition is
placed between two half-step stochastic propagators. We write
\begin{equation}
    U_R = U_L=
    \exp\left(
    \frac{\tau}{2}A_{n+\frac12}
    \right),
\end{equation}
so that every local operation takes the form
\[
    U_L\,(\text{local operation})\,U_R .
\]
In implementation, one may use any symmetric choice of the two half-step
propagators, which is sufficient to preserve second-order accuracy.

For a given auxiliary state \(\psi_n^\sigma\) at time \(t_n\), the
second-order diagrammatic update is built from the following elementary
transitions:

\begin{itemize}
    
    \item \textbf{Propagation:} Contribution to $\psi_{n+1}^{\sigma}$.
\begin{flalign}&
\begin{alignedat}{2}
    &\text{Rule}\qquad
    && \psi^{\sigma}_{n+1}
    \;\leftarrow\; U_L U_R \,\psi^{\sigma}_n
    \\
    &\text{Example}&&\\[-1ex]
    &\quad \sigma=\emptyset:\qquad
    && \adjustbox{valign=c}{\includegraphics[scale=0.6]{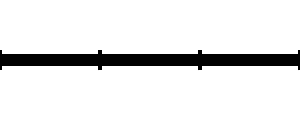}}
    \;=\; \adjustbox{valign=c}{\includegraphics[scale=0.6]{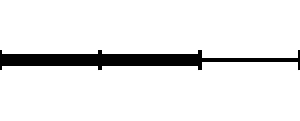}}
    + \cdots
\end{alignedat}&&
\end{flalign}
    The configuration \(\sigma\) is unchanged, while the auxiliary state is
    propagated through the two half-step propagators.

    \item \textbf{Insertion:} New circles generated at \(t_{n+\frac12}\).
    \begin{itemize}
        \item[(i)] \textbf{Single insertion:} Contribution to
        \(\psi_{n+1}^{\sigma\cup(t_{n+\frac12})}\).
    \begin{flalign}&
    \begin{alignedat}{2}
        &\text{Rule}\qquad
        && \psi_{n+1}^{\sigma\cup(t_{n+\frac12})}
        \;\leftarrow\;
        U_L \left[ L_{n+\frac12} \right] U_R \,\psi^{\sigma}_n
        \\
        &\text{Example}&&\\[-1ex]
        &\quad \sigma=\emptyset:\qquad
        && \adjustbox{valign=c}{\includegraphics[scale=0.6]{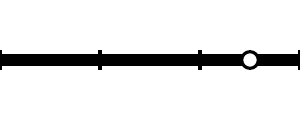}}
        \;=\; \adjustbox{valign=c}{\includegraphics[scale=0.6]{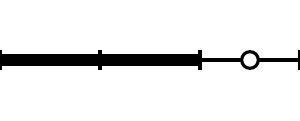}}
        + \cdots
    \end{alignedat}&&
    \end{flalign}
    This is the direct contribution generated by a single insertion of the
current label \(t_{n+\frac12}\). Additional contributions involving this
inserted label arise from memory pairings described below.

                \item[(ii)] \textbf{Double insertion:} Contribution to
        \(\psi_{n+1}^{\sigma\cup(t_{n+\frac12},t_{n+\frac12})}\).
        \begin{flalign}&
        \begin{alignedat}{2}
            &\text{Rule}\qquad
            && \psi_{n+1}^{\sigma\cup(t_{n+\frac{1}{2}}, t_{n+\frac{1}{2}})}
            \;=\; U_L \left[ L^2_{n+\frac{1}{2}} \right] U_R \,\psi^{\sigma}_n
            \\
            &\text{Example}&&\\[-1ex]
            &\quad \sigma=\emptyset:\qquad
            && \adjustbox{valign=c}{\includegraphics[scale=0.6]{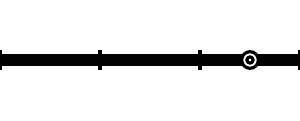}}
            \;=\; \adjustbox{valign=c}{\includegraphics[scale=0.6]{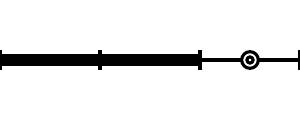}}
        \end{alignedat}&&
        \end{flalign}
                This contribution generates a double circle at the current time,
                when it is admissible in \(\mathcal D_{n+1}^{[2]}\).
    \end{itemize}
    \item \textbf{Pairing:} Memory contractions.
    \begin{itemize}
        \item[(i)] \textbf{Ordinary pairing:} Contribution to
        \(\psi_{n+1}^{\sigma\setminus(s)}\) for each distinct label \(s\in\sigma\).
        \begin{flalign}&
        \begin{alignedat}{2}
            &\text{Rule}\qquad
            && \psi^{\sigma\setminus(s)}_{n+1}
            \;\leftarrow\;
            U_L\left[-\tau^2\alpha_{n+\frac12,s}L^\dagger_{n+\frac12}\right]U_R\psi^\sigma_n
            \\
            &\text{Examples}&&
            \\[-1ex]
            &\sigma=(t_\frac32),\ s= t_\frac32:\quad
            && \adjustbox{valign=c}{\includegraphics[scale=0.6]{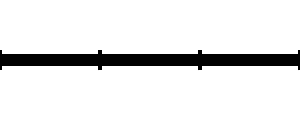}}
            \;=\;
            \adjustbox{valign=c}{\includegraphics[scale=0.6]{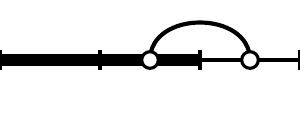}}
            +\cdots
            \\[-3ex]
            &\sigma=(t_\frac32,t_\frac32),\ s= t_\frac32:\quad
            && \adjustbox{valign=c}{\includegraphics[scale=0.6]{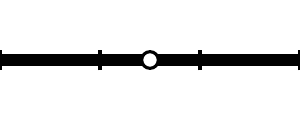}}
            \;=\;
            \adjustbox{valign=c}{\includegraphics[scale=0.6]{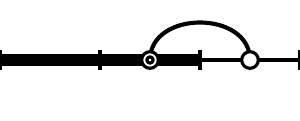}}
            +\cdots
        \end{alignedat}&&
        \end{flalign}
        The label \(s\) is contracted with the current point.
        The resulting contribution is accumulated into the configuration
        \(\sigma\setminus(s)\). If \(s\) appears twice in \(\sigma\), this contraction is counted only once and removes one copy.
        
        \item[(ii)] \textbf{Mixed insertion--pairing:} Contribution to
        \(\psi_{n+1}^{\sigma\setminus(s)\cup(t_{n+\frac12})}\) for each
        \(s\in\sigma\).
        \begin{flalign}&
        \begin{alignedat}{2}
            &\text{Rule}\qquad
            && \psi_{n+1}^{\sigma\setminus(s)\cup(t_{n+\frac12})}
            \;\leftarrow\;
            U_L\left[-\tau^2\alpha_{n+\frac12,s}
            L_{n+\frac12}L^\dagger_{n+\frac12}\right]U_R\psi^\sigma_n
            \\
            &\text{Example}&&
            \\[-1ex]
            &\sigma=(t_\frac32),\ s=t_\frac32:\quad
            && \adjustbox{valign=c}{\includegraphics[scale=0.6]{figure/diagram/add/nmqsd_second_order_def-4.pdf}}
            \;=\;
            \adjustbox{valign=c}{\includegraphics[scale=0.6]{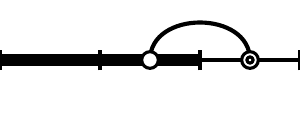}}
            +\cdots
        \end{alignedat}&&
        \end{flalign}
        Here a memory pairing first contracts a past label
        \(s\in\sigma\), and the current label \(t_{n+\frac12}\) is then
        inserted into the remaining configuration.

        \item[(iii)] \textbf{Self-pairing:} Contribution to
        \(\psi_{n+1}^{\sigma}\).
        \begin{equation}
        \begin{alignedat}{2}
            &\text{Rule}\qquad
            && \psi^{\sigma}_{n+1}
            \;\leftarrow\;
            U_L\left[-\frac{\tau^2}{2}\alpha_{n+\frac12,n+\frac12}
            L^\dagger_{n+\frac12}L_{n+\frac12}\right]U_R\psi^\sigma_n
            \\
            &\text{Example}&&
            \\[-1ex]
            &\quad\sigma=\emptyset:\quad
            && \adjustbox{valign=c}{\includegraphics[scale=0.6]{figure/diagram/add/nmqsd_second_order_def-2.pdf}}
            \;=\;
            \adjustbox{valign=c}{\includegraphics[scale=0.6]{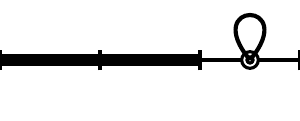}}
            +\cdots
        \end{alignedat}
        \end{equation}
        This term represents the local self-pairing of the current insertion.

    \end{itemize}
\end{itemize}

The algorithmic structure of the second-order diagrammatic scheme is
summarized in Algorithm~\ref{alg:second_order_nmqsd}, and a fully
expanded one-step example is provided in
the supplementary material. In practical implementations,
the distinct labels can still be stored by the same binary representation
as in the first-order scheme, while the possible repeated time label is recorded
separately by an additional marker. The empty marker corresponds to
configurations with distinct labels, and a nonempty marker records the
repeated label.
\begin{algorithm}[htbp]
\caption{Second-Order Hierarchical Evolution Step}
\label{alg:second_order_nmqsd}
\begin{algorithmic}[1]

\Require Auxiliary states
\(\{\psi^\sigma_n\}_{\sigma\in\mathcal D_n^{[2]}}\) at time \(t_n\),
stochastic noise \(z_{n+\frac12}\), time step \(\tau\), correlation
function \(\alpha(t,s)\).

\Ensure Auxiliary states
\(\{\psi^\sigma_{n+1}\}_{\sigma\in\mathcal D_{n+1}^{[2]}}\) at time
\(t_{n+1}\).

\State \textbf{Preprocessing:}
\(
    U_R = \exp\left(\frac{\tau}{2}A_{n+\frac12}\right),
    \qquad
    U_L = \exp\left(\frac{\tau}{2}A_{n+\frac12}\right).
\)

\State Initialize target states:
\(\psi^\sigma_{n+1}= 0\), for all \(\sigma\in\mathcal D_{n+1}^{[2]}\).

\For{each active configuration \(\sigma\in\mathcal D_n^{[2]}\)}

    \State \textbf{1. Propagation and self-pairing:}
    \State
    \[
    \psi^\sigma_{n+1} \mathrel{+\!=} U_L
    \left(
    I-\frac{\tau^2}{2}\alpha_{n+\frac12,n+\frac12}L_{n+\frac12}^\dagger L_{n+\frac12}
    \right)
    U_R\psi^\sigma_n .
    \]

    \State \textbf{2. Insertion:}

    \State \textbf{(i) Single insertion:}
    \(
    \psi^{\sigma\cup(t_{n+\frac12})}_{n+1}
    \mathrel{+\!=}
    U_L L_{n+\frac12} U_R\psi^\sigma_n .
    \)

    \If{\(\sigma\cup(t_{n+\frac12},t_{n+\frac12})
    \in\mathcal D_{n+1}^{[2]}\)}
        \State \textbf{(ii) Double insertion:}
        \(
        \psi^{\sigma\cup(t_{n+\frac12},t_{n+\frac12})}_{n+1}
        \mathrel{+\!=}
        U_L L_{n+\frac12}^2 U_R\psi^\sigma_n .
        \)
    \EndIf

    \State \textbf{3. Pairing:}

    \For{each distinct label \(s\) appearing in \(\sigma\)}

        \State \textbf{(i) Ordinary pairing:}
        \(
        \psi^{\sigma\setminus(s)}_{n+1}
        \mathrel{+\!=}
        U_L
        \left(
        -\tau^2
        \alpha_{n+\frac12,s}
        L_{n+\frac12}^\dagger
        \right)
        U_R\psi^\sigma_n .
        \)

        \State \textbf{(ii) Mixed insertion--pairing:}
        \State
        \[
        \psi^{\sigma\setminus(s)\cup(t_{n+\frac12})}_{n+1}
        \mathrel{+\!=}
        U_L
        \left(
        -\tau^2
        \alpha_{n+\frac12,s}
        L_{n+\frac12}L_{n+\frac12}^\dagger
        \right)
        U_R\psi^\sigma_n .
        \]

    \EndFor

\EndFor

\State \Return
\(\{\psi^\sigma_{n+1}\}_{\sigma\in\mathcal D_{n+1}^{[2]}}\).

\end{algorithmic}
\end{algorithm}

Compared with the first-order scheme, the second-order construction extends
the retained configuration set by allowing one repeated label and places all
local operations between the two symmetric half-step propagations. Together
with the insertion rules for new and repeated labels, these diagrammatic
rules give a closed update for all retained auxiliary states from \(t_n\) to
\(t_{n+1}\).

\section{Numerical Truncation and Complexity}
\label{sec:truncation_complexity}

The discrete auxiliary-state hierarchy constructed above grows rapidly with
the number of time steps. To obtain a tractable numerical method, we introduce
two truncation strategies: memory truncation and hierarchy-level truncation.

\begin{itemize}
    \item \textbf{Memory truncation.}
    Since the bath correlation function \(\alpha(t,s)\) decays as
    \(|t-s|\) increases, the memory contribution can be restricted to a
    finite sliding window. Let \(K_T=K\tau\) denote the memory length, with \(K\) the number of
time steps in this window. At time \(t_n\), only labels satisfying
    \[
        s\in [t_n-K_T,t_n]
    \]
    are retained in the configurations, while older labels are discarded.

    \item \textbf{Hierarchy-level truncation.}
    To further control the number of auxiliary states, we restrict the
    maximum hierarchy level by imposing
    \[
        |\sigma|\le M.
    \]
    Here \(|\sigma|\) denotes the number of unpaired time labels in the
auxiliary configuration. Each memory pairing carries a factor bounded by
\(\tau^2C_\alpha C_L^2\). Hence a configuration with large \(|\sigma|\)
can contribute to lower-level states only after multiple pairings, and is
therefore a higher-order contribution in the discrete memory expansion.
We retain only configurations with \(|\sigma|\le M\).
\end{itemize}

We now estimate the storage and computational cost of the truncated hierarchy.
Let \(\mathcal N_1\) and \(\mathcal N_2\) denote the number of retained
configurations per trajectory for the first- and second-order schemes,
respectively, and let \(N_{\rm traj}\) be the number of stochastic trajectories.
The computational cost per time step is summed over all trajectories.

The storage cost is determined by the number of admissible label
configurations. Suppose that \(q\) time labels are available. In the
first-order scheme, each label is either included or excluded, so the
untruncated configuration count is
\[
   \mathcal{N}_1(q) = \sum_{m=0}^{q}\binom{q}{m}=2^q.
\]
Memory truncation replaces \(q=n\) by the fixed memory depth \(q=K\).
Hierarchy-level truncation further replaces the full subset count by
\[
    \mathcal{N}_1(K,M) = \sum_{m=0}^{M}\binom{K}{m} = \mathcal{O}(K^M),
    \qquad M\ll K.
\]
For the second-order scheme, a distinct configuration with \(m\) labels gives
\(m+1\) retained configurations: the distinct configuration itself and the
\(m\) choices obtained by repeating one of its labels. Thus the same
replacement gives the corresponding second-order count
\[
    \mathcal{N}_2(K,M) = \sum_{m=0}^{M}(m+1)\binom{K}{m} = \mathcal{O}(MK^M),\qquad M\ll K.
\]
The computational cost is obtained by multiplying the number of retained
configurations by the number of transitions generated from each configuration
and by \(N_{\rm traj}\). The resulting storage and computational costs are
summarized in Table~\ref{tab:complexity}.

\begin{table}[htbp]
\centering
\caption{Storage and computational costs before and after truncation.
Here \(n\) is the total number of time steps, \(K\) is the retained memory
depth in time steps, \(M\) is the maximum hierarchy level, and
\(N_{\rm traj}\) is the number of stochastic trajectories. }
\label{tab:complexity}
\begin{tabular}{ccccc}
\toprule
 & Scheme & No trunc. & Trunc. \(K\) & Trunc. \(K,M\) \\
\midrule
\multirow{2}{*}{Storage}
& First-order
& \(\mathcal O(2^n)\)
& \(\mathcal O(2^K)\)
& \(\mathcal O(K^M)\)
\\
& Second-order
& \(\mathcal O(n2^n)\)
& \(\mathcal O(K2^K)\)
& \(\mathcal O(MK^M)\)
\\
\midrule
\multirow{2}{*}{Cost}
& First-order
& \(\mathcal O(N_{\rm traj}n2^n)\)
& \(\mathcal O(N_{\rm traj}K2^K)\)
& \(\mathcal O(N_{\rm traj}MK^M)\)
\\
& Second-order
& \(\mathcal O(N_{\rm traj}n^2 2^n)\)
& \(\mathcal O(N_{\rm traj}K^2 2^K)\)
& \(\mathcal O(N_{\rm traj}M^2K^M)\)
\\
\bottomrule
\end{tabular}
\end{table}
Memory truncation replaces the exponentially growing history of time labels by a fixed memory
depth \(K\). Hierarchy-level truncation further reduces the configuration count
to polynomial scaling in \(K\) for fixed \(M\). Although the second-order scheme
has a larger computational cost per step at the same \(K\) and \(M\), its higher
temporal accuracy can reduce the overall cost when a prescribed accuracy is
required.

The table also clarifies the difference in configuration counts between the
present stochastic hierarchy and density-matrix-based formulations. In a
density-matrix formulation, memory configurations are typically tracked on
both sides of the density matrix. Without truncation, this gives a
first-order configuration count of order
\( \mathcal O(2^{2n})\), and the corresponding update costs \(\mathcal O(n2^{2n})\)
per time step. For long-time evolution, where \(n\) becomes large, the stochastic formulation avoids the additional exponential factor associated with configurations on both sides of the density matrix, although the total cost also scales with the number of stochastic trajectories \(N_{\rm traj}\).

\section{Numerical Results}
\label{sec:results}
In this section, we present numerical experiments to evaluate the performance 
of the proposed hierarchical NMQSD schemes. 
These experiments examine the accuracy, convergence behavior, and
truncation effects of the proposed schemes for non-Markovian open quantum
dynamics.

To quantify the numerical accuracy, we measure the error of observables 
$\langle O \rangle_t$ using the $L^2$-norm
\begin{equation}
e =  
\left\| \langle O \rangle_t 
- \langle O \rangle_t^{\mathrm{ref}} \right\|_{L^2([0,T])},
\end{equation}
where 
\begin{equation}
\|f\|_{L^2([0,T])} = \sqrt{\int_0^T |f(t)|^2 \mathrm{d}t},\quad 
\langle O \rangle_t = \mathrm{Tr}[\rho(t) O] = \mathbb{E}[\langle \psi_t | O | \psi_t \rangle].
\end{equation}
Here $\langle O \rangle^{\mathrm{ref}}$ denotes either an analytical solution, 
when available, or a sufficiently resolved numerical reference solution.
Specific details are provided in section~\ref{sec:supp-experimental-settings}
of the supplementary material.

\subsection{Spin--boson Model}

The spin--boson model serves as a standard benchmark for open quantum systems, 
describing a two-level system coupled to a bosonic environment \cite{leggett_dynamicsdissipativetwostate_1987}.
The system Hamiltonian is 
\begin{equation}
H = \epsilon \sigma_z + \Delta \sigma_x,
\end{equation}
where $\epsilon$ is the energy bias and $\Delta$ is the tunneling amplitude. 
The system--bath coupling operator is $L = \lambda \sigma_z$. 
Here $\sigma_x$, $\sigma_y$, and $\sigma_z$ denote the Pauli matrices:\begin{equation}
\sigma_x = \begin{pmatrix}
    0 & 1\\
    1 & 0
\end{pmatrix},\quad
\sigma_y = \begin{pmatrix}
    0 & -i\\
    i & 0
\end{pmatrix},\quad
\sigma_z = \begin{pmatrix}
    1 & 0\\
    0 & -1
\end{pmatrix}.
\end{equation}

\subsubsection{Exponential Bath}
\label{sec:exp_bath}

We consider an environment with an exponential memory kernel
\(
\alpha(t,s) = \frac{\gamma}{2} e^{-\gamma |t-s|}.
\)
In the pure dephasing limit ($\Delta=0$), this model admits an analytical solution. The density matrix evolves as
\begin{equation}
    \rho(t) = \begin{pmatrix}
        \rho_{11}(0) & \rho_{12}(0)e^{-F(t)}\\
        \rho_{21}(0)e^{-F(t)^*} & \rho_{22}(0)
    \end{pmatrix},
\end{equation}
with
\(
F(t) = 2\mathrm{i}\epsilon t + \frac{2\lambda^2}{\gamma}(e^{-\gamma t} + \gamma t - 1).
\)

\begin{example}[Convergence test with exponential bath]
In this example, we investigate the numerical convergence 
of the proposed numerical schemes. 
The  model parameters are
\begin{equation}
    H = \frac{1}{2}\sigma_z, \quad L = \sqrt{2}\,\sigma_z, \quad \alpha(t,s) = \frac{1}{2} e^{-|t-s|}.
\end{equation}
The initial state is chosen as
\(
\ket{\psi_0} = \frac{1+2i}{\sqrt{7}}\ket{\uparrow} 
+ \frac{1+i}{\sqrt{7}} \ket{\downarrow}.
\)
The reference observables are given by
\begin{equation}
\langle \sigma_x \rangle_t = 2\mathrm{Re}(\rho_{12}(0)e^{-F(t)}),\
\langle \sigma_y \rangle_t = -2\mathrm{Im}(\rho_{12}(0)e^{-F(t)}),\
\langle \sigma_z \rangle_t = \rho_{11}(0) - \rho_{22}(0),
\end{equation}
where $F(t) = \mathrm{i}t + 4(e^{-t} + t - 1)$.
The truncation parameters are set to $K_T = 1$ and $M = 2$.

\label{eg:exp_bath}
\end{example}

For the convergence test, we fix the number of trajectories to 
$N_{\mathrm{traj}} = 10^6$ and vary the time step $\tau$ 
from $0.5$ to $0.1$. The numerical results for the exponential bath model are presented in 
Figs.~\ref{fig:exp_dynamics}--\ref{fig:exp_convergence}.

\begin{figure}[htbp]
    \centering
    \includegraphics[width=1.0\textwidth]{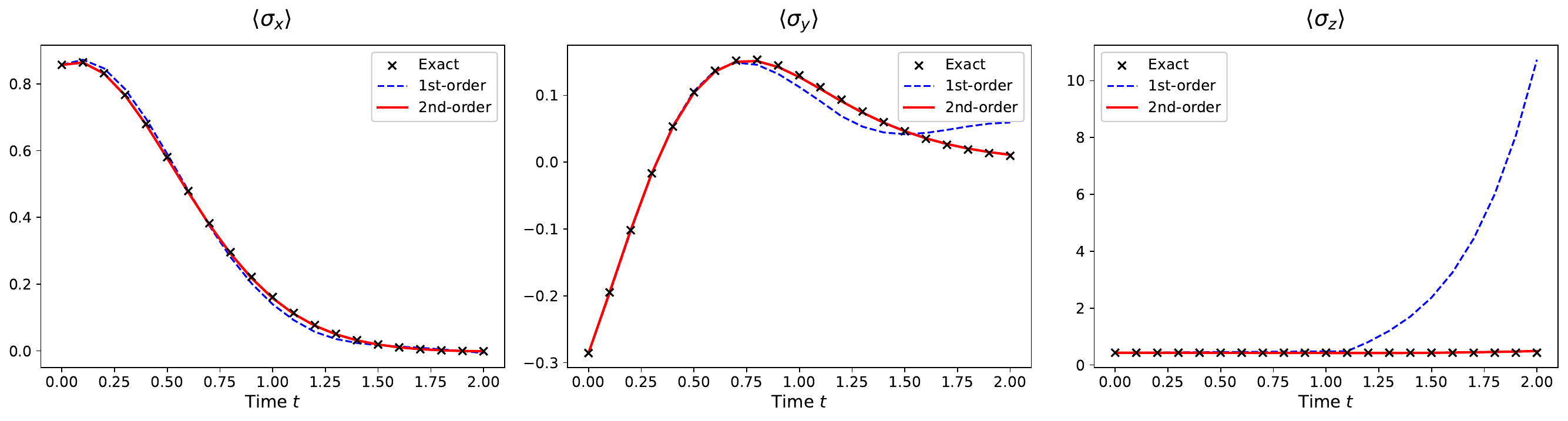}
    \caption{Spin dynamics for the exponential bath (Example \ref{eg:exp_bath}). Time evolution of $\langle \sigma_x \rangle$, $\langle \sigma_y \rangle$, and $\langle \sigma_z \rangle$ with $\tau = 0.1$. The first-order scheme (blue dashed) and the second-order scheme (red solid) are compared with the exact solution (cross markers).}
    \label{fig:exp_dynamics}
\end{figure}

Figure~\ref{fig:exp_dynamics} compares the time evolution of 
$\langle \sigma_x \rangle$, $\langle \sigma_y \rangle$, and $\langle \sigma_z \rangle$ 
for the first- and second-order schemes against the analytical solution. 
At this time step, the first-order scheme shows visible numerical
deviations from the analytical solution, especially in
$\langle \sigma_z \rangle$ and $\langle \sigma_y \rangle$.
In contrast, the second-order scheme agrees well with the reference
dynamics for all three observables, showing improved accuracy at the
same time step.

\begin{figure}[htbp]
    \centering
    \includegraphics[width=0.9\textwidth]{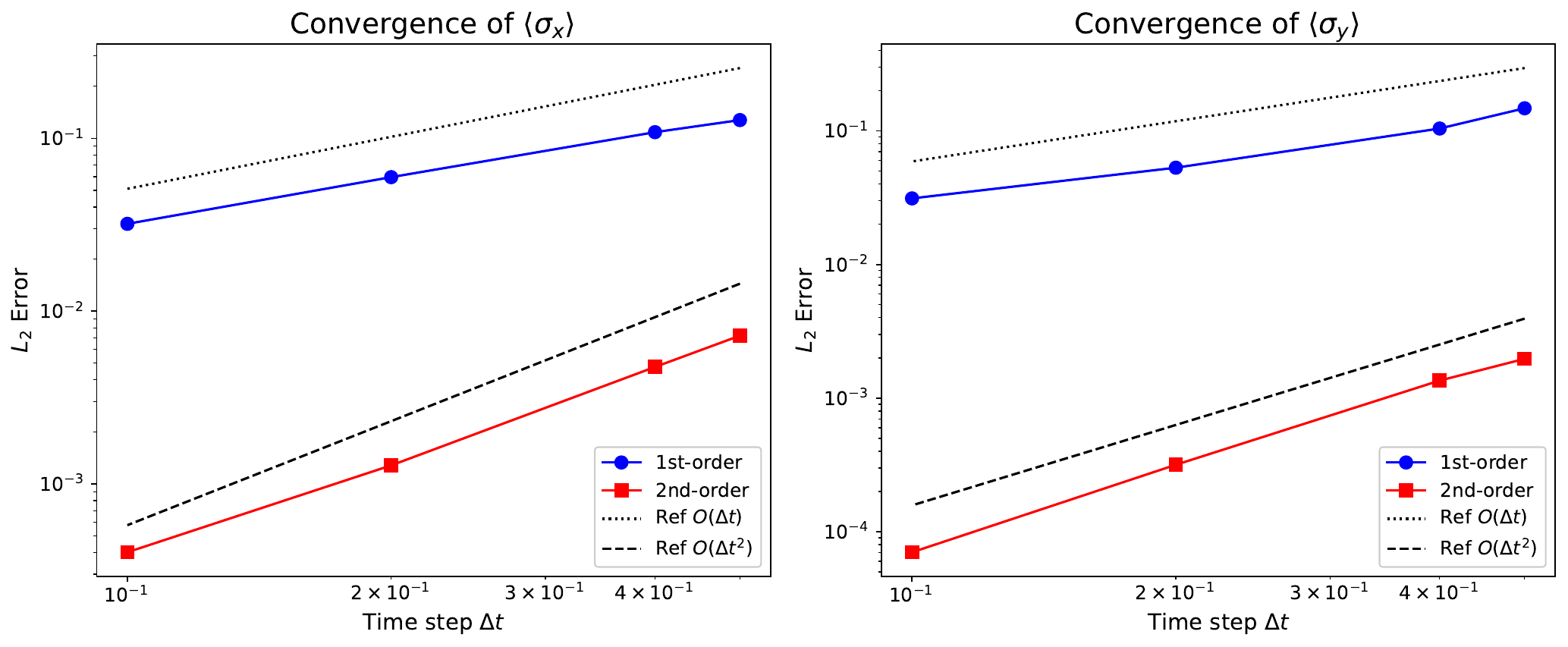}
    \caption{ Convergence test of $\langle \sigma_x \rangle$ and $\langle \sigma_y \rangle$ versus time step $\tau$ (Example \ref{eg:exp_bath}).}
    \label{fig:exp_convergence}
\end{figure}
The convergence behavior is shown in Fig.~\ref{fig:exp_convergence}. 
For $\langle \sigma_x \rangle$ and $\langle \sigma_y \rangle$, 
the first-order scheme achieves observed convergence rates 
of approximately $0.86$ and $0.95$, while the second-order scheme 
gives rates of $1.80$ and $2.09$, respectively. 
These observed rates are broadly consistent with the expected first- and
second-order accuracy.

Moreover, over the tested range of time steps, the second-order scheme
reduces the global error by about two orders of magnitude compared with
the first-order scheme.
Since the computational cost grows rapidly with the memory length $K = K_T/\tau$,
this reduction in time-discretization error is important in practice,
as it allows accurate results to be obtained with relatively coarse time
steps.

Figure~\ref{fig:exp_traj_sensitivity} illustrates the effect of the number of trajectories 
on the ensemble average. 
For the time-dependent observables $\langle \sigma_x \rangle$ and
$\langle \sigma_y \rangle$, $N_{\mathrm{traj}} = 10^3$ already gives
ensemble averages close to the analytical reference.
The constant observable $\langle \sigma_z \rangle$ shows slower convergence,
but the sampling error is substantially reduced when $N_{\mathrm{traj}}$
is increased to $10^5$.
\begin{figure}[htbp]
    \centering
    \includegraphics[width=1.0\textwidth]{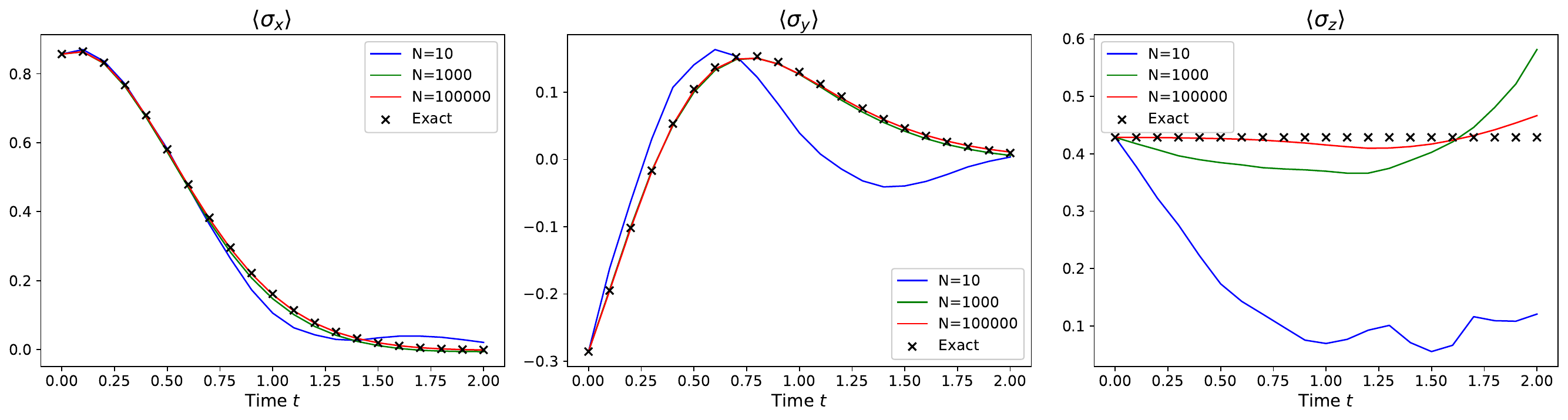}
    \caption{Ensemble averages of spin observables calculated with $N_{\mathrm{traj}} = 10, 10^3, 10^5$ and $\tau = 0.1$ (Example \ref{eg:exp_bath}). The lines show the stochastic averages, and the cross markers denote the analytical reference.}
    \label{fig:exp_traj_sensitivity}
\end{figure}

\subsubsection{Ohmic Bath}\label{sec:ohmic_bath}

Next, we consider a more challenging environment characterized by an Ohmic spectral density, 
whose bath correlation function is difficult to approximate
accurately by a short sum of exponentials.
The bath correlation function is given by
\begin{equation}\alpha(t,s) = \int_0^\infty \frac{J(\omega)}{\pi} \left[\coth\left(\frac{\beta\omega}{2}\right)\cos(\omega(t-s)) - \mathrm{i}\sin(\omega(t-s))\right] \mathrm{d} \omega,\end{equation}
where $J(\omega) = \frac{\pi}{2} \xi\omega e^{-\omega/\omega_c}$ denotes the Ohmic spectral density with cutoff frequency  $\omega_c$ \cite{caldeira_quantumtunnellingdissipative_1983}.  
For numerical simulation, the continuous spectrum is discretized using $L = 400$ modes, 
\begin{equation}
    J(\omega) \approx \frac{\pi}{2}\sum_{l=1}^L \frac{c_l^2}{\omega_l} \delta(\omega - \omega_l),
\end{equation}
where the coupling strengths $c_l$ and frequencies $\omega_l$ are given by
\begin{equation}c_l = \omega_l\sqrt{\frac{\xi \omega_c}{L}(1-e^{-\omega_{\max}/\omega_c})}, \quad \omega_l = -\omega_c \ln\left[1-\frac{l}{L}(1-e^{-\omega_{\max}/\omega_c})\right].
\end{equation}
Here $\xi$ is the dimensionless coupling strength, and  
$\omega_{\max} = 4\omega_c$.

\begin{example}[Dynamics and convergence test with Ohmic bath]
    In this example, we simulate the non-Markovian dynamics of the spin-boson model with an 
    Ohmic bath. 
    The system is defined by
    \begin{equation}
        \begin{aligned}
                    H &= \epsilon \sigma_z + \Delta \sigma_x, \qquad L = \sigma_z, 
        \\ \alpha(t,s) &= \sum_{j=1}^L \frac{c_j^2}{2\omega_j}\left(
\coth\left(\frac{\beta\omega_j}{2}\right)\cos(\omega_j(t-s)) - \mathrm{i}\sin(\omega_j(t-s))
        \right)
        \end{aligned}
    \end{equation}
    The initial state is chosen as $\ket{\psi_0} = (1, 0)^\top$.
Under the parameter settings $\omega_c = 2.5\Delta$, $\xi = 0.2$, and $\beta = 5\Delta^{-1}$, we simulate the population difference $\langle \sigma_z \rangle_t$ for energy biases $\epsilon \in \{0, \Delta, 2\Delta\}$.
The truncation parameters are set to $K_T = 1$ and $M = 2$.
    \label{ex:ohmic}
\end{example}

The convergence behavior is shown in Fig.~\ref{fig:ohmic_conv}, 
where the $L^2$ error of $\langle \sigma_z \rangle_t$ 
is plotted against the time step $\tau$. 
In all cases, the error curves follow the reference slope
$\mathcal{O}(\tau^2)$, indicating the expected second-order convergence
behavior for the ensemble-averaged observable.
This result shows that the proposed scheme is accurate for the tested
Ohmic bath.

\begin{figure}[htbp]
    \centering
    \includegraphics[scale=0.3]{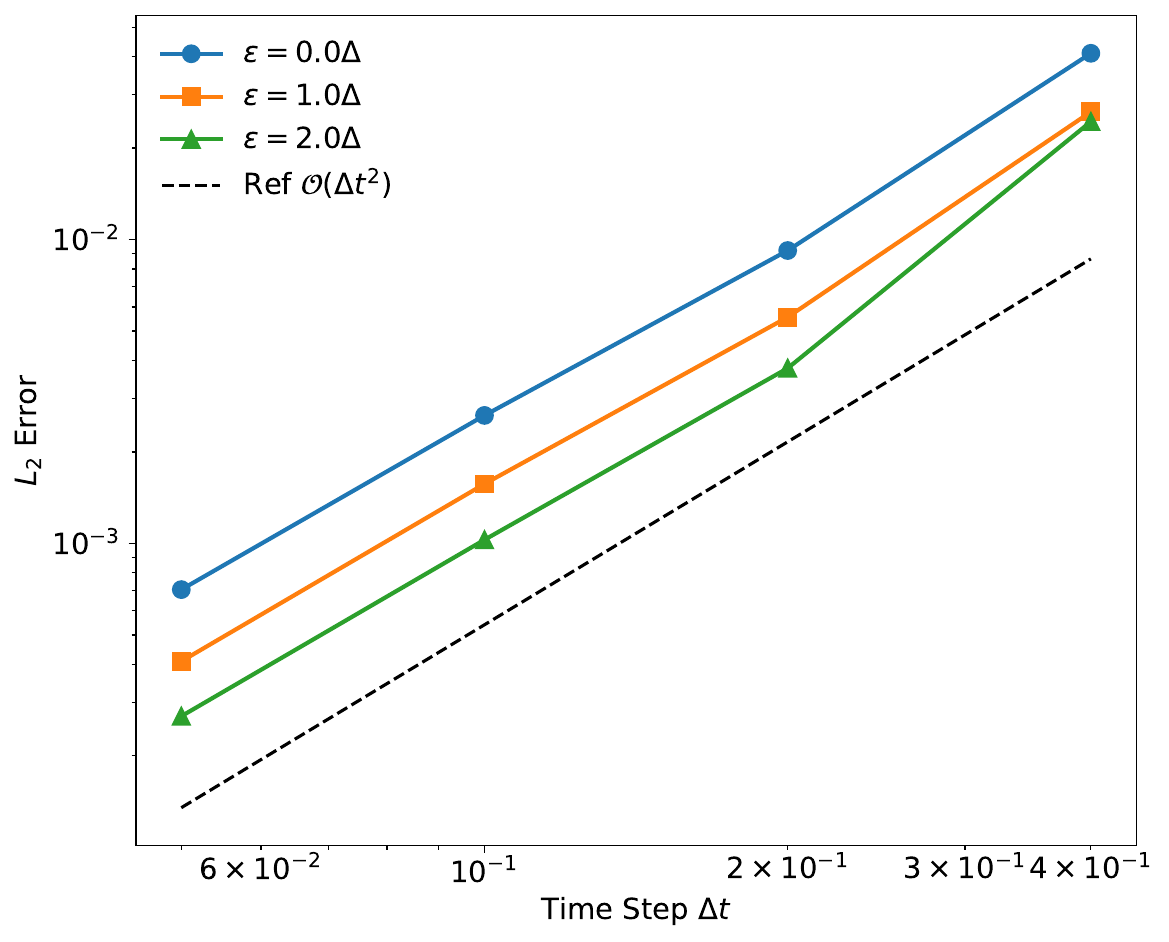}
    \caption{The $L_2$ error of $\langle \sigma_z \rangle_t$ versus the time step $\tau$ for Example \ref{ex:ohmic} with $\epsilon \in \{0, \Delta, 2\Delta\}$. The dashed line indicates the reference slope $\mathcal{O}(\tau^2)$.}
    \label{fig:ohmic_conv}
\end{figure}

\begin{figure}[htbp]
    \centering
    \includegraphics[width=1.0\textwidth]{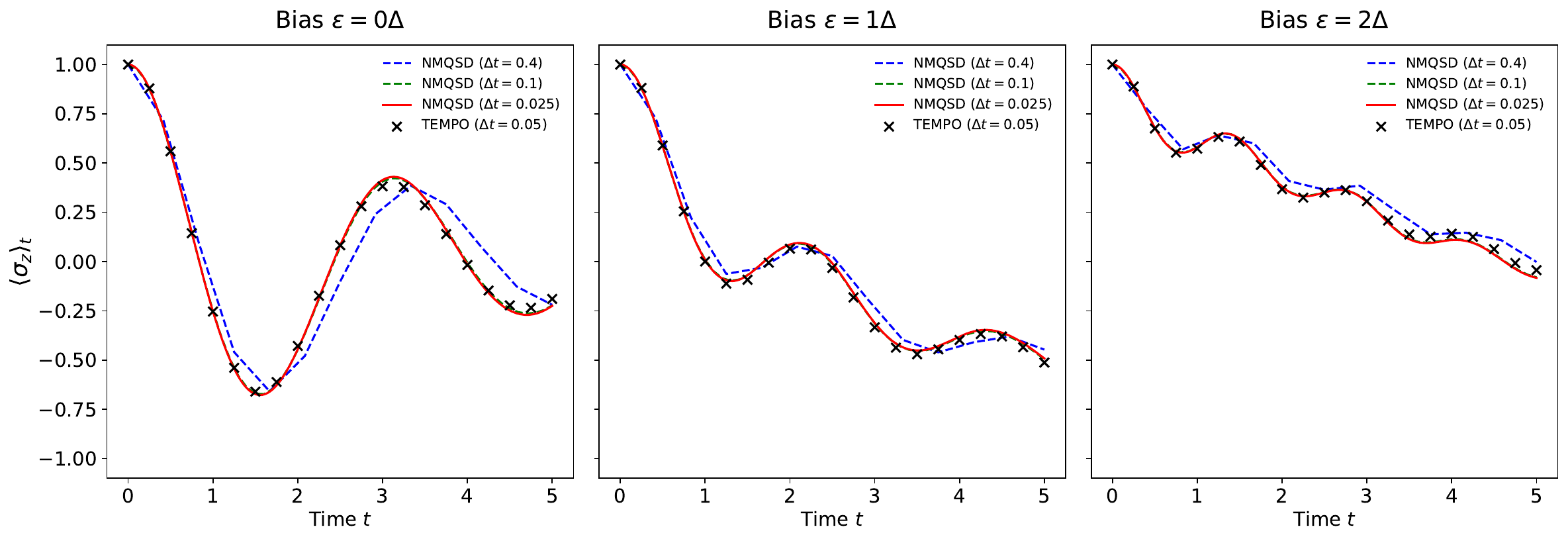}
    \caption{Time evolution of $\langle \sigma_z \rangle_t$ 
for different energy biases in Example \ref{ex:ohmic}. 
The NMQSD solutions with $\tau \in \{0.4, 0.1, 0.025\}$ 
are compared with the TEMPO reference solution.}
    \label{fig:ohmic_dyn}
\end{figure}

Figure~\ref{fig:ohmic_dyn} shows the time evolution of 
$\langle \sigma_z \rangle_t$ up to $T = 5.0$. 
The NMQSD results obtained with $\tau = 0.4$, $0.1$, and $0.025$ 
are compared with a TEMPO reference solution \cite{strathearn_efficientnonmarkovianquantum_2018,fux_oqupypythonpackage_2024}.
The TEMPO reference is generated using the \texttt{OQuPy} package with time step $\tau_{\mathrm{ref}} = 0.05$ and relative SVD truncation tolerance $10^{-7}$.

For the coarse time step $\tau = 0.4$, visible phase and amplitude errors are observed. 
In contrast, the results with $\tau = 0.1$ and $0.025$ agree well with the TEMPO reference, indicating that the second-order scheme captures the dynamics accurately at these time steps.

\begin{example}[Truncation sensitivity under different coupling strengths]
This example investigates the sensitivity of the numerical results to the truncation
parameters under different system--bath coupling strengths.  
We fix the energy bias $\epsilon = 0$ and the time step $\tau = 0.1$, use $N_{\mathrm{traj}}=10^4$ trajectories, and consider $\xi \in \{0.2,0.4,0.6,0.8\}$.

For each value of $\xi$, we vary the memory length $K$ and the maximum
hierarchy level $M$.
For the weaker couplings $\xi=0.2$ and $0.4$, we use
$K\in\{9,12,15,18\}$, while for the stronger couplings $\xi=0.6$ and $0.8$,
we use $K\in\{18,21,24,27\}$. The hierarchy level is varied
over $M\in\{1,2,3,4,5\}$.
All other parameters are the same as in Example~\ref{ex:ohmic}.

\label{ex:xi_sensitivity}
\end{example}

We first examine the dependence on the maximum hierarchy level $M$, as
shown in Fig.~\ref{fig:M_convergence}.  The memory length is fixed at
$K=18$ for $\xi=0.2,0.4$ and at $K=27$ for $\xi=0.6,0.8$.  For weak
coupling, the dynamics is already stable at $M=2$.  For stronger coupling,
lower hierarchy levels differ visibly, whereas the curves for the two
largest tested values of $M$ nearly overlap.  This indicates internal
convergence with respect to the hierarchy truncation, with higher hierarchy
levels required as the coupling becomes stronger.

\begin{figure}[htbp]
    \centering
    \includegraphics[width=0.8\textwidth]{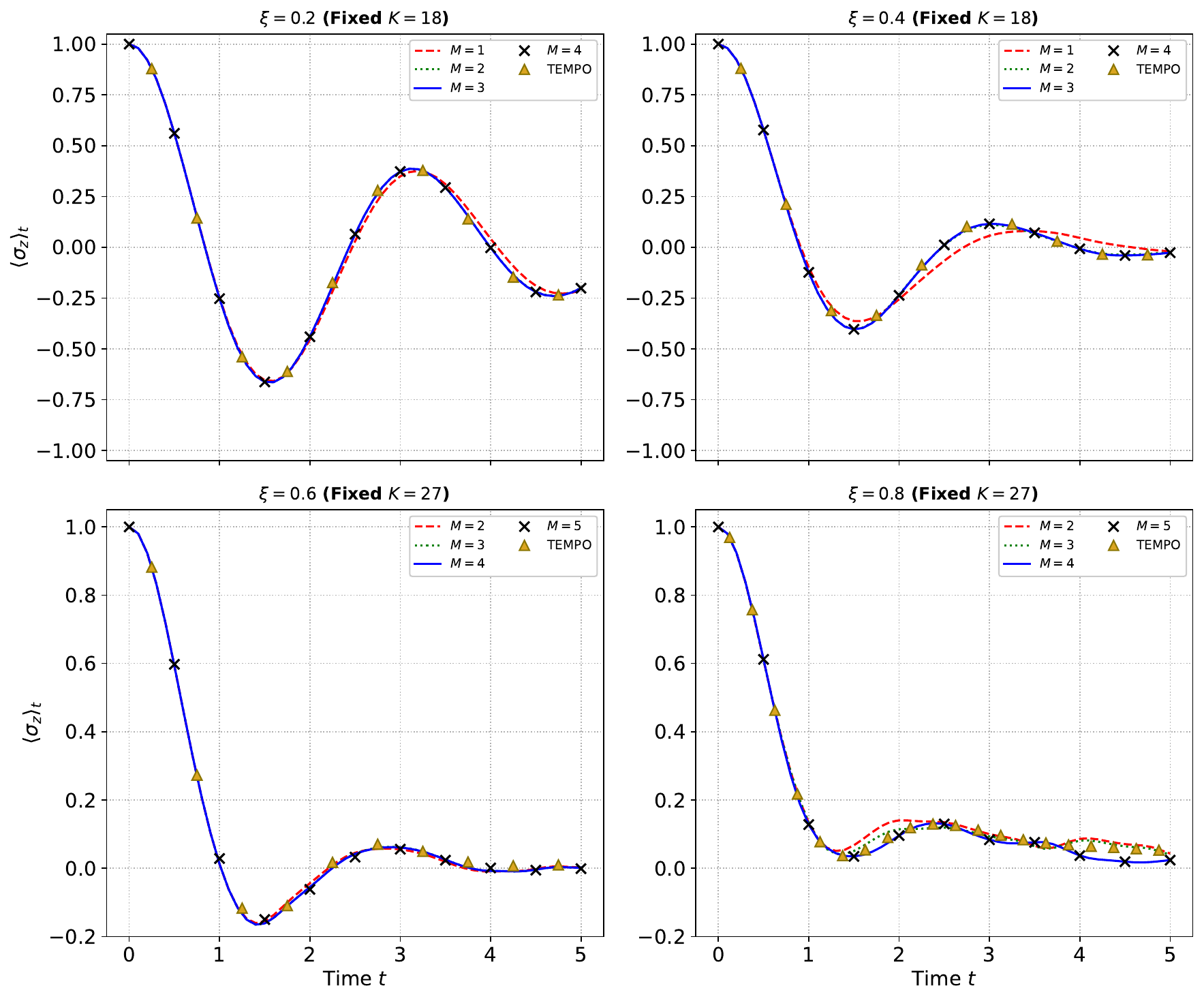}
    \caption{Sensitivity of the dynamics $\langle \sigma_z \rangle_t$ to the hierarchy level $M$ 
    for various coupling strengths $\xi$ in Example \ref{ex:xi_sensitivity}. The memory length is fixed at $K=18$ for $\xi=0.2,0.4$ and at $K=27$ for $\xi=0.6,0.8$, with $\tau=0.1$ and
$N_{\mathrm{traj}}=10^4$. }    \label{fig:M_convergence}
\end{figure}

Figure~\ref{fig:K_convergence} shows the sensitivity to the memory length
$K$ with fixed $M$.
For weak coupling, increasing $K$ beyond $12$ produces little change in the dynamics.
For stronger coupling, memory effects persist over
a longer time window: the curves still change noticeably up to $K=21$, but
the results for $K=24$ and $K=27$ are nearly indistinguishable.  This suggests that the memory-truncation effect is well controlled within the
tested range of $K$.

\begin{figure}[htbp]
    \centering
    \includegraphics[width=0.8\textwidth]{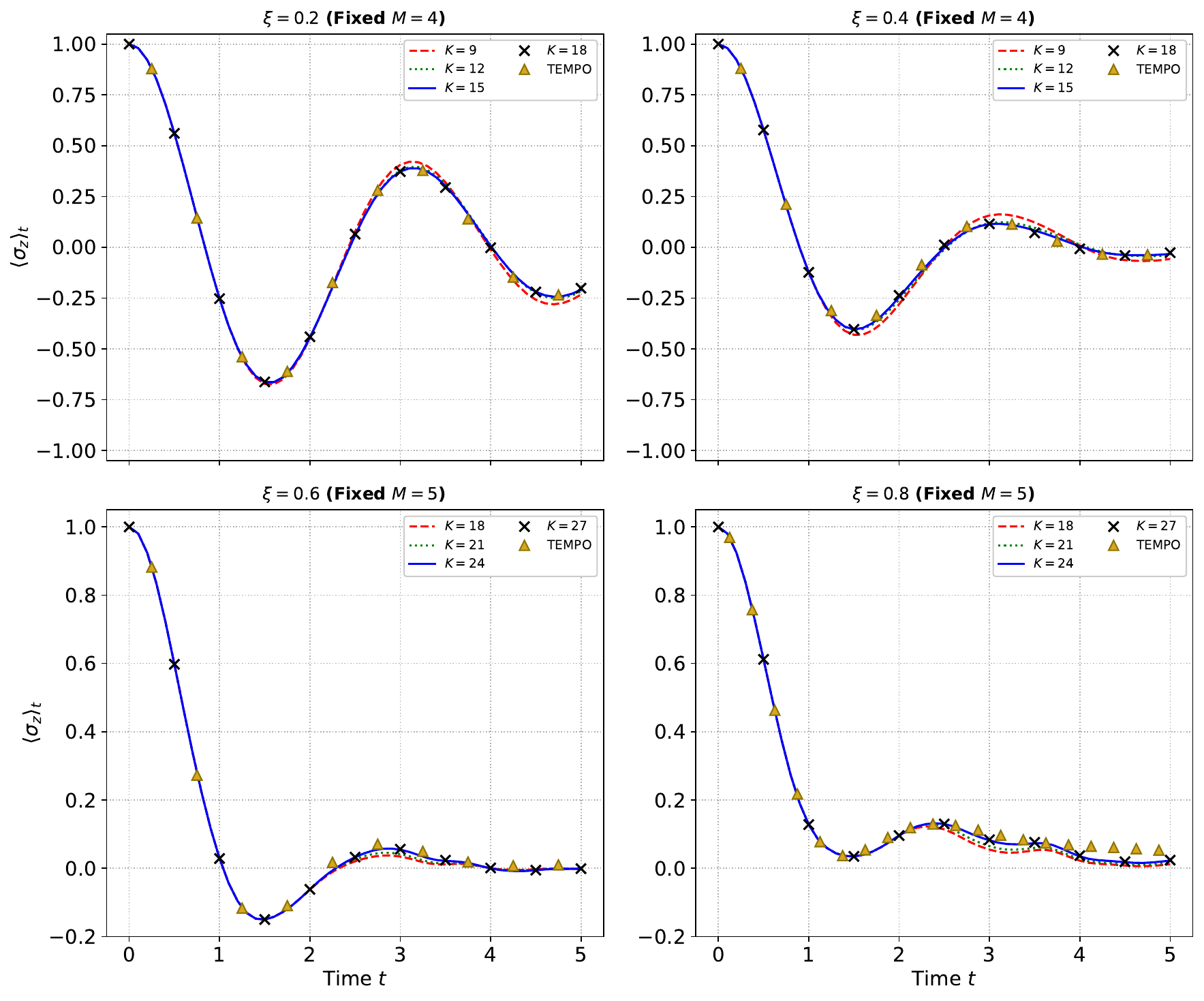}
    \caption{Sensitivity of the dynamics $\langle \sigma_z \rangle_t$ to the memory length $K$ for various coupling strengths $\xi$ in Example \ref{ex:xi_sensitivity}. We use $K=9,12,15,18$ and $M=4$ for $\xi=0.2,0.4$, and $K=18,21,24,27$ and $M=5$ for $\xi=0.6,0.8$, with
$\tau=0.1$ and $N_{\mathrm{traj}}=10^4$.}
    \label{fig:K_convergence}
\end{figure}

The strong-coupling cases also show that agreement with the TEMPO reference
is not strictly monotone in the truncation parameters.  A lower hierarchy
level can appear closer to TEMPO in some long-time regions, but the
convergence tests indicate that the NMQSD dynamics has already stabilized
with respect to both $M$ and $K$ at the largest tested values.  The remaining
long-time deviation is therefore better viewed as a residual discrepancy
between the two numerical approaches. In this sense, the present NMQSD scheme
provides an independent cross-check of strongly coupled non-Markovian
dynamics.

\subsection{Multi-level System}
In this example, we apply the proposed method to a higher-dimensional
system to illustrate its applicability beyond two-level models.
\begin{example}[Multi-level systems]
Consider a $d$-level system with the system Hamiltonian defined as
\begin{equation}
    H = \sum_{i=1}^{d-1} (|i\rangle\langle i+1| + |i+1\rangle\langle i|).
\end{equation}
The system--bath coupling operator is 
\begin{equation}
    L = \mathrm{diag}\left(-1, -1+\frac{2}{d-1}, \dots, 1\right).
\end{equation}
The environment is modeled by the Ohmic spectral density used in
Example~\ref{ex:ohmic}. We take \(d=11\), \(\omega_c=2.5\), $\xi=0.2$ and \(\beta=1.0\).
We compare the following two initial system states: 
\( (i)\ \rho_s(0)=|1\rangle\langle 1|,\ (ii)\ \rho_s(0) = \frac{1}{2}\left(|1\rangle\langle 1|+|d\rangle\langle d|\right).
\) 
The evolution is monitored through the level populations \[ P_i(t)=\langle i|\rho_s(t)|i\rangle . \] 
In all simulations, the time step is set to \(\tau=0.1\), the number of trajectories is \(N_{\mathrm{traj}}=10^4\), and the truncation parameters are \(K=18\) and \(M=3\).
\label{ex:multilevel}
\end{example}
Fig.~\ref{fig:multilevel_dynamics} shows the population dynamics in the
high-temperature regime, represented here by the inverse temperature
\(\beta=1\), for two different initial states.
For the localized initial state (i), the population
starts from the first level and is gradually redistributed across the level
chain. For the symmetric mixed initial state (ii),
the redistribution starts from both boundary levels and gives a nearly
symmetric population profile. These results illustrate that the proposed second-order NMQSD scheme can be naturally applied to higher-dimensional systems without changing the algorithmic structure. 

\begin{figure}[htbp]
    \centering
    \includegraphics[width=0.9\textwidth]{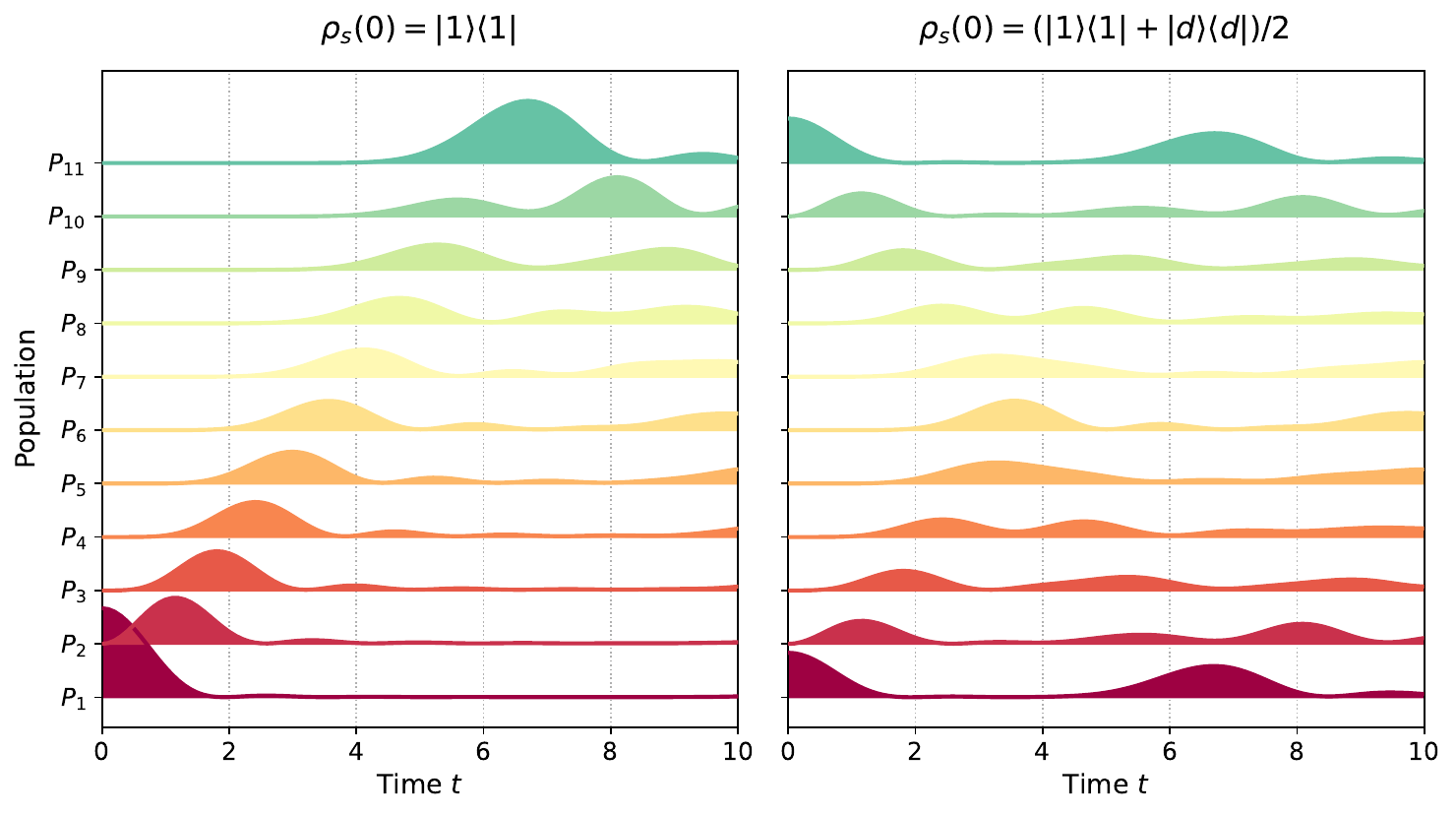}
\caption{Time evolution of the level populations \(P_i(t)\) for the \(d=11\) system in Example~\ref{ex:multilevel} in the high-temperature regime with \(\beta=1.0\). Left: localized initial state \(\rho_s(0)=|1\rangle\langle 1|\). Right: symmetric mixed initial state \(\rho_s(0)=\frac{1}{2}(|1\rangle\langle 1|+|d\rangle\langle d|)\). }
    \label{fig:multilevel_dynamics}
\end{figure}

\section{Conclusion}
\label{sec:conclusion}

This work develops a general numerical framework for the
linear NMQSD equation with arbitrary bath correlation functions. The
construction is based on an analytical representation that separates the
non-Markovian stochastic dynamics into propagation, insertion, and
memory pairing. This structure provides a direct way to represent
the functional derivatives through auxiliary states, without requiring a
prescribed decomposition of the bath correlation function.

The framework organizes the numerical construction into time discretization,
memory quadrature, and hierarchy truncation. Within this formulation, we
derive first- and second-order schemes and implement them through
diagrammatic update rules. The memory and hierarchy-level truncations make
the resulting auxiliary-state hierarchy finite and computationally
controllable.

Numerical experiments verify the expected temporal accuracy of the proposed
schemes and demonstrate their applicability to different bath correlation
functions and multi-level quantum systems. 
These results show that the proposed framework leads to general, accurate,
and computationally tractable numerical methods for non-Markovian stochastic
dynamics.

Future work includes extensions to many-body open quantum systems, more
general system--bath coupling structures, adaptive truncation strategies, and
higher-order auxiliary-state discretizations.

\clearpage

\renewcommand{\siamprelabel}{SM}
\setcounter{section}{0}
\setcounter{subsection}{0}
\setcounter{subsubsection}{0}
\setcounter{equation}{0}
\setcounter{theorem}{0}
\setcounter{figure}{0}
\setcounter{table}{0}
\setcounter{algorithm}{0}
\renewcommand{\theHsection}{supp.\arabic{section}}
\renewcommand{\theHsubsection}{\theHsection.\arabic{subsection}}
\renewcommand{\theHsubsubsection}{\theHsubsection.\arabic{subsubsection}}
\renewcommand{\theHequation}{supp.\arabic{section}.\arabic{equation}}
\renewcommand{\theHtheorem}{supp.\arabic{section}.\arabic{theorem}}
\renewcommand{\theHfigure}{supp.\arabic{figure}}
\renewcommand{\theHtable}{supp.\arabic{table}}
\renewcommand{\theHalgorithm}{supp.\arabic{algorithm}}

\newcounter{supplementstartpage}
\setcounter{supplementstartpage}{\value{page}}
\renewcommand{\thepage}{SM\the\numexpr\value{page}-\value{supplementstartpage}+1\relax}

\headers{Supplementary Materials: General Numerical Solver for NMQSD}{Z. Cai and Q. Zhu}
\begin{center}
{\bfseries
\MakeUppercase{Supplementary Materials: A General First- and Second-Order
Numerical Solver for Non-Markovian Quantum State Diffusion}\par}
\vskip .075in
{\footnotesize ZHENNING CAI \scriptsize AND \footnotesize QUANHUI ZHU\par}
\end{center}
\vskip .11in

\adjustboxset{valign=m, height=2em, raise=-0.2ex}
This supplementary material contains the proofs of
Theorems~2.1, 3.2, and 3.3, together with fully expanded examples of the
first- and second-order hierarchical updates.

\section{Proof of Theorem~\ref{thm:solution}}
\label{app:proof_solution}
To facilitate the differentiation with respect to time,
it is convenient to rewrite the integration over the ordered simplex
$\Delta_n(t)$ as an integral over the full cube $C_n = [0,t]^n$.
Using the symmetry of the integrand under permutations of
$\boldsymbol{\tau}$, we obtain
\begin{equation}
    \psi_t^{(n,m)}
    = (-1)^m
    \int_{C_n} \frac{1}{n!} \D^n \boldsymbol{\tau}
    \int_{\Delta_{2m}(t)} f(\boldsymbol{\tau}, \boldsymbol{s}) \D^{2m} \boldsymbol{s}.
    \label{eq:solution_cube}
\end{equation}
This representation is equivalent to the simplex formulation,
but allows for a more convenient treatment of the time derivative.
We now verify that the series represented by \eqref{eq:solution_cube}
satisfies the NMQSD equation.

First, we differentiate an individual term $\psi_t^{(n,m)}$ with respect to time.
Using the equivalent cube representation \eqref{eq:solution_cube} of analytical solution,
the time dependence enters only through the integration domains.
Applying the Leibniz rule yields
\begin{equation}
    \partial_t \psi_t^{(n,m)} = A_{n,m} + B_{n,m},
\end{equation}
where
\begin{itemize}
    \item $A_{n,m}$ is the contribution from the boundary of the cube
    $C_n = [0,t]^n$, corresponding to $\tau_k = t$ for some $k$.
    \item $B_{n,m}$ is the contribution from the boundary of the simplex
    $\Delta_{2m}(t)$, corresponding to $s_{2m} = t$.
\end{itemize}

\paragraph{1. Stochastic driving term $A_{n,m}$}

For $n\ge 1$, the symmetry of the integrand under permutations of
$\boldsymbol{\tau}$ over $C_n$ implies that all boundary contributions
with $\tau_k = t$ are identical.
Fixing $\tau_n = t$ and summing over the $n$ equivalent cases, we obtain
\begin{equation}
\begin{aligned}
    A_{n,m}
    &= n \cdot (-1)^m
    \int_{C_{n-1}} \frac{1}{n!}\D^{n-1}\boldsymbol{\tau}
    \int_{\Delta_{2m}(t)} f_A(\boldsymbol{\tau},\boldsymbol{s}) \D^{2m}\boldsymbol{s} \\
    &= (-1)^m
    \int_{C_{n-1}} \frac{1}{(n-1)!}\D^{n-1}\boldsymbol{\tau}
    \int_{\Delta_{2m}(t)} f_A(\boldsymbol{\tau},\boldsymbol{s}) \D^{2m}\boldsymbol{s},
\end{aligned}
\end{equation}
where
\begin{equation}
    f_A = z_t^* \left( \prod_{k=1}^{n-1} z_{\tau_k}^* \right)
    \mathcal{L}_b(\boldsymbol{s})
    \mathcal{T}\!\left(
    L(t)\prod_{k=1}^{n-1}L(\tau_k)
    \prod_{j=1}^m L^\dagger(s_{P(2j)})L(s_{P(2j-1)})
    \right)\psi_0.
\end{equation}
Since $t$ is the largest time, $L(t)$ is placed
leftmost by the time-ordering operator. Therefore,
\begin{equation}
    A_{n,m} = z_t^* L(t)\,\psi_t^{(n-1,m)}.
    \label{eq:Anm_result_fixed}
\end{equation}

\paragraph{2. Memory term $B_{n,m}$}
For $m\ge 1$, the term $B_{n,m}$ is the contribution from the boundary
of the ordered simplex with $s_{2m} = t$. Applying the Leibniz rule gives
\begin{equation}
\begin{aligned}
B_{n,m}
&= (-1)^m
\int_{C_n} \frac{1}{n!}\D^n \boldsymbol{\tau}
\int_{\Delta_{2m-1}(t)}
f_B(\boldsymbol{\tau},\boldsymbol{s})
\D^{2m-1}\boldsymbol{s},
\end{aligned}
\end{equation}
where
\begin{equation}
\begin{aligned}
f_B(\boldsymbol{\tau},\boldsymbol{s})
&=
\left( \prod_{k=1}^n z_{\tau_k}^* \right)
\mathcal{L}_b(s_1,\dots,s_{2m-1},t)
\\
&\quad\times
\mathcal{T}\!\left(
\prod_{k=1}^n L(\tau_k)
\prod_{j=1}^{m-1}
\big[L^\dagger(s_{P(2j)}) L(s_{P(2j-1)})  \big]
\, L^\dagger(t)\, L(s_{P(2m-1)})
\right)\psi_0.
\end{aligned}
\end{equation}
In the pairing structure $\mathcal{L}_b(s_1,\dots,s_{2m-1},t)$,
the time $t$ must be paired with one of the variables
$s_l$, $l=1,\dots,2m-1$.
Thus,
\begin{equation}
\mathcal{L}_b(s_1,\dots,s_{2m-1},t)
=
\sum_{l=1}^{2m-1}
\alpha(t,s_l)\,
\mathcal{L}_b(s_1,\dots,s_{l-1},s_{l+1},\dots,s_{2m-1}).
\end{equation}

Substituting this decomposition into $B_{n,m}$ yields
\begin{equation}
\begin{aligned}
B_{n,m}
&=
(-1)^m \sum_{l=1}^{2m-1}
\int_{C_n} \frac{1}{n!}\D^n \boldsymbol{\tau}
\int_{\Delta_{2m-1}(t)}
\alpha(t,s_l)\,
f_B(\boldsymbol{\tau},\boldsymbol{s},l)
\D^{2m-1}\boldsymbol{s},
\end{aligned}
\end{equation}
where
\begin{equation}
\begin{aligned}
f_B(\boldsymbol{\tau},\boldsymbol{s},l)
&=
\left( \prod_{k=1}^n z_{\tau_k}^* \right)
\mathcal{L}_b(s_1,\dots,s_{l-1},s_{l+1},\dots,s_{2m-1})
\\
&\quad\times
\mathcal{T}\!\left(
\prod_{k=1}^n L(\tau_k)
\prod_{j=1}^{m-1}
\big[ L^\dagger(s_{P(2j)}) L(s_{P(2j-1)})  \big]
 L^\dagger(t)  L(s_l)
\right)\psi_0.
\end{aligned}
\end{equation}

\paragraph{3. Functional derivative term}

We now evaluate the functional-derivative contribution on the right-hand side
of \eqref{eq:NMQSD_reduced}. For $m\ge 1$, in
$\psi_t^{(n+1,m-1)}$, the functional derivative acts on the product of
stochastic variables as
\begin{equation}
\frac{\delta}{\delta z_s^*}
\prod_{k=1}^{n+1} z_{\tau_k}^*
=
\sum_{l=1}^{n+1}
\delta(\tau_l - s)
\prod_{k\neq l} z_{\tau_k}^*.
\end{equation}
This gives
\begin{equation}
\begin{aligned}
\frac{\delta \psi_t^{(n+1,m-1)}}{\delta z_s^*}
&=
(-1)^{m-1}
\int_{C_{n+1}} \frac{1}{(n+1)!}\D^{n+1}\boldsymbol{\tau}\int_{\Delta_{2m-2}(t)} 
\sum_{l=1}^{n+1}
\delta(\tau_l - s) \left( \prod_{k\neq l} z_{\tau_k}^* \right)
\\
&\times
\mathcal{L}_b(\boldsymbol{s})
\mathcal{T}\!\left(
\prod_{k=1}^{n+1} L(\tau_k)
\prod_{j=1}^{m-1}
\big[ L^\dagger(s_{P(2j)}) L(s_{P(2j-1)})  \big]
\right)\psi_0
\D^{2m-2}\boldsymbol{s}.
\end{aligned}
\end{equation}
Using the identity 
\begin{equation}
\int_0^t \alpha(t,s)\delta(\tau_l - s)\D s
=
\alpha(t,\tau_l),
\end{equation}
and changing the order of integration, we evaluate the memory integral as

\begin{equation}
\begin{aligned}
I & = -L^\dagger(t)\int_0^t \alpha(t,s)
\frac{\delta \psi_t^{(n+1,m-1)}}{\delta z_s^*}\D s
\\
&=
(-1)^{m}
\sum_{l=1}^{n+1}
\int_{C_{n+1}} \frac{1}{(n+1)!}\alpha(t,\tau_l)\D^{n+1}\boldsymbol{\tau}
\int_{\Delta_{2m-2}(t)} f_M(\boldsymbol{\tau},\boldsymbol{s},l) \D^{2m-2}\boldsymbol{s},
\end{aligned}
\end{equation}
where 
\begin{equation}
    f_M(\boldsymbol{\tau},\boldsymbol{s},l) = 
\left( \prod_{k\neq l} z_{\tau_k}^* \right)
\mathcal{L}_b(\boldsymbol{s})
\mathcal{T}\!\left(L^\dagger(t)
\prod_{k=1}^{n+1} L(\tau_k)
\prod_{j=1}^{m-1}
\big[L^\dagger(s_{P(2j)}) L(s_{P(2j-1)})  \big]
\right)\psi_0.
\end{equation}

It remains to isolate the selected variable $\tau_l$.
Using the symmetry under permutations of
$\boldsymbol{\tau}$ over $C_{n+1}$, we fix this variable as
$\tau_{n+1}$ and sum over the $(n+1)$ equivalent choices, yielding
\begin{equation}
\begin{aligned}
    I
    &= (-1)^m
    \int_{C_n} \frac{1}{n!}\D^n\boldsymbol{\tau}
    \int_0^t \D \tau'
    \alpha(t,\tau')\int_{\Delta_{2m-2}(t)} \hat{f}_M(\boldsymbol{\tau},\boldsymbol{s},\tau') \D^{2m-2}\boldsymbol{s},
\end{aligned}
\end{equation}
where $\tau'=\tau_{n+1}$, and $\hat{f}_M$ denotes $f_M$ after this
relabeling of the selected variable.

The new variable $\tau'$ is then inserted into the ordered simplex
$\Delta_{2m-2}(t)$ of the pairing variables.
The domain $[0,t]\times \Delta_{2m-2}(t)$ can be decomposed into
$2m-1$ ordered regions according to the relative position of $\tau'$
among $(s_1,\dots,s_{2m-2})$:
\begin{equation}
    [0,t]\times \Delta_{2m-2}(t)
    = \bigsqcup_{l=1}^{2m-1} \Delta_{2m-1}^{(l)}(t),
\end{equation}
where $\Delta_{2m-1}^{(l)}(t)$ corresponds to inserting $\tau'$
between $s_{l-1}$ and $s_l$, with the conventions $s_0=0$ and
$s_{2m-1}=t$.

After relabeling $\tau'$ as $s_l$ in the $l$-th region, we obtain
\begin{equation}
    I
    = (-1)^m 
    \int_{C_n} \frac{1}{n!}\D^n\boldsymbol{\tau}
    \int_{\Delta_{2m-1}(t)}\sum_{l=1}^{2m-1}
    \alpha(t,s_l)\, f_B(\boldsymbol{\tau},\boldsymbol{s},l)
    \D^{2m-1}\boldsymbol{s}.
\end{equation}

Therefore,
\begin{equation}
    -L^\dagger(t)\int_0^t \alpha(t,s)
    \frac{\delta \psi_t^{(n+1,m-1)}}{\delta z_s^*}\D s
    = B_{n,m}.
\end{equation}

\paragraph{4. Conclusion.}
Using the conventions
\begin{equation}
    \psi_t^{(-1,m)}=\psi_t^{(n,-1)}=0,
    \qquad
    \frac{\delta \psi_t^{(0,m)}}{\delta z_s^*}=0,
\end{equation}
and summing over all $n,m\ge 0$, we obtain
\begin{equation}
\begin{aligned}
    \partial_t \psi_t
    &= \sum_{n,m\ge 0} (A_{n,m} + B_{n,m}) \\
    & = \sum_{n,m\ge 0} z_t^* L(t)\psi_t^{(n-1,m)} - \sum_{n,m\ge 0}L^\dagger(t)\int_0^t \alpha(t,s)
    \frac{\delta \psi_t^{(n+1,m-1)}}{\delta z_s^*}\D s \\
    &= z_t^* L(t)\psi_t
    - L^\dagger(t)\int_0^t \alpha(t,s)
    \frac{\delta \psi_t}{\delta z_s^*}\D s,
\end{aligned}
\end{equation}
which is precisely the NMQSD equation.
\qed

\section{Proof of Theorem~\ref{thm:error_first_order}}
\label{app:first_order_consistency}
The proof is carried out under Assumptions~\ref{assumption} stated in
Section~\ref{sec:methods} of the main manuscript.

We compare the first-order scheme
\begin{equation}
    \begin{aligned}
    \psi_{n+1}^\sigma & = \psi_n^\sigma + \tau A_n\psi_n^\sigma + \tau B_n^{[1]} \psi_n^\sigma, \\
    \psi_{n+1}^{\sigma\cup(t_n)} & =L_n\psi_{n}^\sigma, \\
    \end{aligned}
    \label{eq:first_order_distinct_update_app}
\end{equation}
with the full Euler hierarchy 
\begin{equation}
    \begin{aligned}
    \phi_{n+1}^\sigma & = \phi_n^\sigma + \tau A_n\phi_n^\sigma + \tau B_n \phi_n^\sigma, \\
    \phi_{n+1}^{\sigma\cup(t_n)} & =L_n\phi_{n}^\sigma, \\
    \end{aligned}
    \label{eq:first_order_full_hierarchy_app}
\end{equation}
and show that the retained scheme preserves the first-order accuracy of the
physical state.

\begin{proof}
We compare the two schemes on the common distinct-index configurations.  For \(0\le k\le n\), define the level-\(k\) subset of
\(\mathcal C_n^{[1]}\) by
\begin{equation}
    \Sigma_n^{(k)}
    :=
    \left\{
    \sigma\in\mathcal C_n^{[1]}:
    |\sigma|=k
    \right\},
    \qquad
    \Sigma_n^{(0)}:=\{\emptyset\}.
    \label{eq:admissible_set_level_k}
\end{equation}
Then \(\mathcal C_n^{[1]}=\bigcup_{k=0}^{n}\Sigma_n^{(k)}\).
Although the full solution \(\phi_n^\sigma\) also includes
repeated-label configurations, only its components on \(\Sigma_n^{(k)}\) are
compared with the implemented solution. For \(\sigma\in\Sigma_n^{(k)}\), define
\begin{equation}
    e_n^\sigma:=\psi_n^\sigma-\phi_n^\sigma .
\end{equation}
The level-wise error is
\begin{equation}
    E_n^{(k)}
    :=
    \max_{\sigma\in\Sigma_n^{(k)}}\|e_n^\sigma\|,
    \qquad 0\le k\le n .
    \label{eq:level_error_admissible}
\end{equation}
For notational convenience, we set
\(E_n^{(k)}=0\) for $k>n$. At \(n=0\), only the physical index is admissible, and
\[
    E_0^{(0)}
    =
    \|\psi_0^\emptyset-\phi_0^\emptyset\|
    =
    0 .
\]

We first estimate the error for configurations that are already
admissible at time \(t_n\).  Let \(\sigma\in\Sigma_n^{(k)}\).  Subtracting the full update
\eqref{eq:first_order_full_hierarchy_app} from the implemented update
\eqref{eq:first_order_distinct_update_app} gives
\begin{equation}
\begin{aligned}
    e_{n+1}^{\sigma}
    &=e_n^\sigma + \tau A_ne_n^\sigma
    -\tau^2L_n^\dagger\sum_{\substack{0\le j\le n-1\\ t_j\notin\sigma}}
    \alpha_{n,j}e_n^{\sigma\cup(t_j)}
    \\
    &+\tau^2L_n^\dagger
    \sum_{\substack{0\le j\le n-1\\ t_j\in\sigma}}
    \alpha_{n,j}\phi_n^{\sigma\cup(t_j)} .
\end{aligned}
\label{eq:error_recursion_repeated_indices_admissible}
\end{equation}
Whenever \(t_j\notin\sigma\), one has
\[
    \sigma\cup(t_j)\in\Sigma_n^{(k+1)}.
\]
Define the error on already admissible configurations by
\begin{equation}
    \widetilde E_{n+1}^{(k)}
    :=
    \max_{\sigma\in\Sigma_n^{(k)}}\|e_{n+1}^{\sigma}\|.
    \label{eq:old_history_error}
\end{equation}
For \(0\le k\le n\), taking norms in
\eqref{eq:error_recursion_repeated_indices_admissible} and using the
assumptions gives
\begin{equation}
    \widetilde E_{n+1}^{(k)}
    \le
    (1+\tau C_A)E_n^{(k)}
    +
    \tau C_LC_{\alpha,T}E_n^{(k+1)}
    +
    D_n^{(k)},
    \label{eq:old_history_level_recursion}
\end{equation}
where
\begin{equation}
    D_n^{(k)}
    :=
    \max_{\sigma\in\Sigma_n^{(k)}}
    \left\|
    \tau^2L_n^\dagger
    \sum_{\substack{0\le j\le n-1\\ t_j\in\sigma}}
    \alpha_{n,j}
    \phi_n^{\sigma\cup(t_j)}
    \right\|.
    \label{eq:defect_Dnk}
\end{equation}
Since a distinct configuration \(\sigma\in\Sigma_n^{(k)}\) contains exactly \(k\)
time points, the coefficient and level-wise bounds imply
\begin{equation}
\begin{aligned}
    D_n^{(k)}
    &\le
    \tau^2 C_L
    \sum_{\substack{0\le j\le n-1\\ t_j\in\sigma}}
    |\alpha_{n,j}|\,
    \|\phi_n^{\sigma\cup(t_j)}\|
    \\
    &\le
    \tau^2 C_L k C_\alpha C_\phi M^{k+1}.
\end{aligned}
\label{eq:levelwise_defect_bound_admissible}
\end{equation}
In particular, \(D_n^{(0)}=0\).

Next we account for the configurations that are newly admissible at time
\(t_{n+1}\).  For \(1\le k\le n+1\),
\begin{equation}
    \Sigma_{n+1}^{(k)}
    =
    \Sigma_n^{(k)}
    \cup
    \left\{
    \eta\cup(t_n):\eta\in\Sigma_n^{(k-1)}
    \right\},
    \label{eq:admissible_set_decomposition}
\end{equation}
with the convention that \(\Sigma_n^{(k)}=\emptyset\) when \(k>n\).
The configurations in the second set are generated by the insertion rule.  
Since both schemes use the same rule, for
\(\eta\in\Sigma_n^{(k-1)}\),
\begin{equation}
    \psi_{n+1}^{\eta\cup(t_n)}
    =
    L_n\psi_n^\eta,
    \qquad
    \phi_{n+1}^{\eta\cup(t_n)}
    =
    L_n\phi_n^\eta .
    \label{eq:discrete_initialization_rule_error}
\end{equation}
Hence
\begin{equation}
    e_{n+1}^{\eta\cup(t_n)}
    =
    L_n e_n^\eta .
    \label{eq:initialization_error_rule}
\end{equation}

Choose a fixed \(r>0\) such that
\begin{equation}
    rM<1,
    \qquad
    rC_L\le 1 .
    \label{eq:r_choice}
\end{equation}
Define the weighted errors by
\begin{equation}
    \mathcal E_n
    :=
    \max_{0\le k\le n} r^k E_n^{(k)},
    \qquad
    \widetilde{\mathcal E}_{n+1}
    :=
    \max_{0\le k\le n} r^k \widetilde E_{n+1}^{(k)} .
    \label{eq:weighted_admissible_error}
\end{equation}
Multiplying \eqref{eq:old_history_level_recursion} by \(r^k\) and taking
the maximum over \(0\le k\le n\), we obtain
\begin{equation}
\begin{aligned}
    \widetilde{\mathcal E}_{n+1}
    &\le
    (1+\tau C_A)\mathcal E_n
    +
    \tau C_LC_{\alpha,T}
    \max_{0\le k\le n} r^kE_n^{(k+1)}
    +
    \max_{0\le k\le n} r^kD_n^{(k)} .
\end{aligned}
\label{eq:weighted_old_history_prebound}
\end{equation}
By the convention \(E_n^{(k)}=0\) for \(k>n\), we have
\begin{equation}
    \max_{0\le k\le n}r^kE_n^{(k+1)}=
    \max_{0\le k\le n-1}r^k E_n^{(k+1)}
    =
    \frac1r
    \max_{1\le m\le n} r^mE_n^{(m)}
    \le
    \frac1r\mathcal E_n .
\end{equation}
Moreover, by \eqref{eq:levelwise_defect_bound_admissible},
\begin{equation}
\begin{aligned}
    \max_{0\le k\le n}r^kD_n^{(k)}
    &\le
    \tau^2 C_LC_\alpha C_\phi M
    \max_{1\le k\le n} k(rM)^k
    \\
    &\le
    \tau^2 C_LC_\alpha C_\phi M
    \sum_{k=1}^{\infty} k(rM)^k
    \\
    &=
    \tau^2 C_LC_\alpha C_\phi M
    \frac{rM}{(1-rM)^2}.
\end{aligned}
\end{equation}
Set
\begin{equation}
    D_r
    :=
    C_LC_\alpha C_\phi M
    \frac{rM}{(1-rM)^2},
    \qquad
    \Lambda_r
    :=
    C_A+\frac{C_LC_{\alpha,T}}{r}.
\end{equation}
Then
\begin{equation}
    \widetilde{\mathcal E}_{n+1}
    \le
    (1+\tau\Lambda_r)\mathcal E_n+\tau^2D_r .
    \label{eq:old_weighted_error_bound}
\end{equation}

We next pass from the configurations already present at time \(t_n\)
to the full admissible set at time \(t_{n+1}\). By the decomposition
\eqref{eq:admissible_set_decomposition}, each admissible configuration at
\(t_{n+1}\) is either an existing configuration in \(\Sigma_n^{(k)}\), or a newly
generated configuration of the form \(\eta\cup(t_n)\) with
\(\eta\in\Sigma_n^{(k-1)}\). Hence
\begin{equation}
\begin{aligned}
    \mathcal E_{n+1}
    &=
    \max_{0\le k\le n+1}
    r^k
    \max_{\sigma\in\Sigma_{n+1}^{(k)}}
    \|e_{n+1}^\sigma\|
    \\
    &\le
    \max\left\{
    \widetilde{\mathcal E}_{n+1},
    \max_{1\le k\le n+1}
    \max_{\eta\in\Sigma_n^{(k-1)}}
    r^k\|e_{n+1}^{\eta\cup(t_n)}\|
    \right\}.
\end{aligned}
\label{eq:full_vs_old_weighted_error}
\end{equation}
For a newly generated configuration, \eqref{eq:initialization_error_rule} gives
\begin{equation}
\begin{aligned}
    r^k\|e_{n+1}^{\eta\cup(t_n)}\|
    &\le
    r^k C_L\|e_n^\eta\|
    \\
    &=
    rC_L\, r^{k-1}\|e_n^\eta\|
    \\
    &\le
    \mathcal E_n ,
\end{aligned}
\label{eq:newborn_weighted_error}
\end{equation}
where we used \(rC_L\le1\). Therefore,
\begin{equation}
    \mathcal E_{n+1}
    \le
    \max\left\{
    \widetilde{\mathcal E}_{n+1},
    \mathcal E_n
    \right\}.
    \label{eq:weighted_error_full_old_relation}
\end{equation}
Combining this with \eqref{eq:old_weighted_error_bound}, and using
\[
    (1+\tau\Lambda_r)\mathcal E_n+\tau^2D_r\ge \mathcal E_n,
\]
we obtain 
\begin{equation}
    \mathcal E_{n+1}
    \le
    (1+\tau\Lambda_r)\mathcal E_n+\tau^2D_r .
    \label{eq:weighted_gronwall_repeated_admissible}
\end{equation}
Since the physical initial states agree, \(\mathcal E_0=0\).

Applying the discrete Gr\"onwall inequality to
\eqref{eq:weighted_gronwall_repeated_admissible} yields
\begin{equation}
\begin{aligned}
    \mathcal E_n
    &\le
    \tau^2D_r
    \sum_{\ell=0}^{n-1}(1+\tau\Lambda_r)^\ell
    \\
    &\le
    \tau D_r
    \frac{e^{\Lambda_r t_n}-1}{\Lambda_r}
    \le
    \tau D_r
    \frac{e^{\Lambda_r T}-1}{\Lambda_r}.
\end{aligned}
\end{equation}
Finally, since \(\emptyset\in\Sigma_n^{(0)}\), we have
\begin{equation}
    \|\psi_n^\emptyset-\phi_n^\emptyset\|
    =
    E_n^{(0)}
    \le
    \mathcal E_n .
\end{equation}
Therefore,
\begin{equation}
    \max_{0\le n\le N}
    \|\psi_n^\emptyset-\phi_n^\emptyset\|
    \le
    C_T\tau ,
\end{equation}
where
\begin{equation}
    C_T
    :=
    D_r
    \frac{e^{\Lambda_r T}-1}{\Lambda_r}.
\end{equation}
This proves the claim.
\end{proof}

\section{Example of the First-Order Scheme}
\label{sec:first_order_example}
To illustrate the structure of the hierarchical updates,
we present a fully expanded example of the first-order scheme at $t_2$.
At this stage, the existing configurations are subsets of $\{t_0,t_1\}$,
and the insertion rule generates new configurations containing $t_2$.

Each term in the expansion corresponds to one of the three fundamental operations:
propagation, insertion, and pairing.
The diagrammatic notation provides a visual representation of how these operations
generate and connect different configurations in the hierarchy.
\begin{align*}
    |\psi^{\emptyset}_2\rangle &= 
    \adjustbox{}{\includegraphics{figure/diagram/first_order/nmqsd_first_order-1.pdf}} \notag \\
    &= \adjustbox{}{\includegraphics{figure/diagram/first_order/nmqsd_first_order-9.pdf}} +
       \adjustbox{}{\includegraphics{figure/diagram/first_order/nmqsd_first_order-11.pdf}} +
       \adjustbox{}{\includegraphics{figure/diagram/first_order/nmqsd_first_order-12.pdf}} \\
    |\psi^{\{t_0\}}_2\rangle &= 
    \adjustbox{}{\includegraphics{figure/diagram/first_order/nmqsd_first_order-4.pdf}} \notag \\
    &= \adjustbox{}{\includegraphics{figure/diagram/first_order/nmqsd_first_order-15.pdf}} +
       \adjustbox{}{\includegraphics{figure/diagram/first_order/nmqsd_first_order-16.pdf}} \\
    |\psi^{\{t_1\}}_2\rangle &= 
    \adjustbox{}{\includegraphics{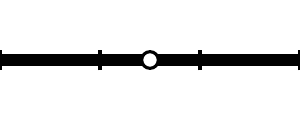}} \notag \\
    &= \adjustbox{}{\includegraphics{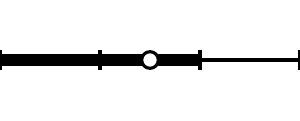}} +
       \adjustbox{}{\includegraphics{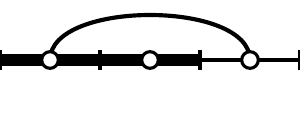}} \\
    |\psi^{\{t_2\}}_2\rangle &= 
    \adjustbox{}{\includegraphics{figure/diagram/first_order/nmqsd_first_order-2.pdf}} 
    = \adjustbox{}{\includegraphics{figure/diagram/first_order/nmqsd_first_order-10.pdf}} \\
    |\psi^{\{t_0, t_1\}}_2\rangle &= 
    \adjustbox{}{\includegraphics{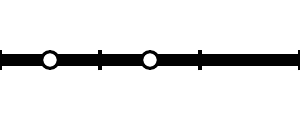}} 
    = \adjustbox{}{\includegraphics{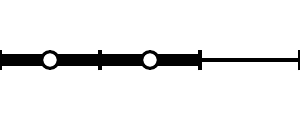}} \\
    |\psi^{\{t_0, t_2\}}_2\rangle &= 
    \adjustbox{}{\includegraphics{figure/diagram/first_order/nmqsd_first_order-6.pdf}} 
    = \adjustbox{}{\includegraphics{figure/diagram/first_order/nmqsd_first_order-18.pdf}} \\
    |\psi^{\{t_1, t_2\}}_2\rangle &= 
    \adjustbox{}{\includegraphics{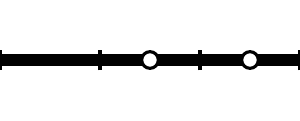}} 
    = \adjustbox{}{\includegraphics{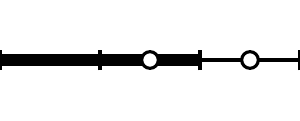}} \\
    |\psi^{\{t_0, t_1, t_2\}}_2\rangle &= 
    \adjustbox{}{\includegraphics{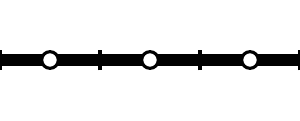}} 
    = \adjustbox{}{\includegraphics{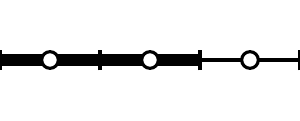}}
\end{align*}

\section{Proof of Theorem~\ref{thm:second_order_retained_memory_consistency}}
\label{app:second_order_consistency}
Recall that the second-order scheme is defined by
\begin{equation}
\begin{alignedat}{3}
    \text{\bfseries Right half-step:}\qquad
    &\hat{\psi}_{n+\frac12}^{\sigma}
    &&=
    \exp\left(\frac{\tau}{2}A_{n+\frac12}\right)
    \psi_n^\sigma,
    &&\qquad \sigma\in\mathcal D_n^{[2]},
    \\
    \text{\bfseries Insertion:}\qquad
    &\hat{\psi}_{n+\frac12}^{\sigma\cup(t_{n+\frac12})}
    &&=
    L_{n+\frac12}
    \hat{\psi}_{n+\frac12}^{\sigma},
    &&\qquad \sigma\in\mathcal D_n^{[2]},
    \\
    &
    \hat{\psi}_{n+\frac12}^{\sigma\cup(t_{n+\frac12},t_{n+\frac12})}
    &&=
    L_{n+\frac12}^{2}
    \hat{\psi}_{n+\frac12}^{\sigma},
    &&\qquad \sigma\in\mathcal D_n^{[2]},
    \\
    \text{\bfseries Memory pairing:}\qquad
    &\widetilde{\psi}_{n+\frac12}^\sigma
    &&=
    \left(
    I+\tau B_{n+\frac12}^{[2]}
    \right)
    \hat{\psi}_{n+\frac12}^\sigma,
    &&\qquad \sigma\in\mathcal D_{n+1}^{[2]},
    \\
    \text{\bfseries Left half-step:}\qquad
    &\psi_{n+1}^{\sigma}
    &&=
    \exp\left(\frac{\tau}{2}A_{n+\frac12}\right)
    \widetilde{\psi}_{n+\frac12}^{\sigma},
    &&\qquad \sigma\in\mathcal D_{n+1}^{[2]},
\end{alignedat}
\label{eq:second_order_appendix_update_summary}
\end{equation}
where \(\mathcal D_n^{[2]}\) is the set of retained configurations at time step
\(n\), and \(\mathcal C_n^{[2]}\) is the distinct-time configuration set. We denote by
\(\phi\) the corresponding full second-order hierarchy, obtained by replacing
the retained memory-pairing step with
\[    
    \widetilde{\phi}_{n+\frac12}^{\sigma} = (I+\tau B_{n+\frac12})\hat\phi_{n+\frac12}^{\sigma}.
\]

The proof is carried out under Assumptions~\ref{assumption} stated in
Section~\ref{sec:methods} of the main manuscript.

\begin{proof}
Let
\begin{equation}
    U_{n+\frac12}:=\exp\Bigl(\frac{\tau}{2}A_{n+\frac12}\Bigr),
\qquad
\Delta B_{n+\frac12}:=
B_{n+\frac12}-B^{[2]}_{n+\frac12}.
\end{equation}
Although \(\phi_n^\sigma\) is defined on the full discrete hierarchy, the error
will only be compared on the common retained configurations. Thus, for
\(\sigma\in\mathcal D_n^{[2]}\), define
\begin{equation}
e_n^\sigma:=\psi_n^\sigma-\phi_n^\sigma,\qquad \hat{e}_{n+\frac12}^\sigma:=\hat{\psi}_{n+\frac12}^\sigma-\hat{\phi}_{n+\frac12}^\sigma.
\end{equation}
For \(k\ge0\), set
\begin{equation}
\Sigma_n^{(k)}:=\{\sigma\in\mathcal D_n^{[2]}:\ |\sigma|=k\},
\qquad
\Gamma_n^{(k)}:=\{\sigma\in\mathcal C_n^{[2]}:\ |\sigma|=k\},
\end{equation}
where \(|\sigma|\) denotes the total number of labels counted with
multiplicity. Here \(\Gamma_n^{(k)}\) is the distinct-label configuration set of level \(k\), while
\(\Sigma_n^{(k)}\) contains both distinct configurations and configurations
with one repeated label. 
The empty configuration \(\sigma=\emptyset\) corresponds to the
physical state, and therefore
\[
    \Sigma_n^{(0)}=\Gamma_n^{(0)}=\{\emptyset\}.
\]
Since \(\mathcal C_n^{[2]}\) contains \(n\) labels, no distinct
configuration of level \(k>n\) is admissible. Similarly, since
\(\mathcal D_n^{[2]}\) allows at most one repeated label, no retained
configuration of level \(k>n+1\) is admissible. We use the convention that
the maximum over such an empty admissible set is zero.

Choose constants $0<r_1<r_0$ such that
\[
r_0M<1,\qquad r_0C_L\le1.
\]
Define
\begin{equation}
E_n:=
\max_{k\ge0} r_0^k
\max_{\sigma\in\Sigma_n^{(k)}}\|e_n^\sigma\|,
\qquad
F_n:=
\max_{k\ge0} r_1^k
\max_{\sigma\in\Gamma_n^{(k)}}\|e_n^\sigma\|.
\end{equation}
At $n=0$, only the empty configuration is admissible, and the two physical initial
states agree. Hence
\[
E_0=F_0=0.
\]

We first estimate \(E_n\). After the right half-step,
\begin{equation}
\widehat e_{n+\frac12}^\sigma
=
U_{n+\frac12}e_n^\sigma ,
\qquad \sigma\in\mathcal D_n^{[2]}.
\end{equation}
The boundedness of \(A(t)\) implies the half-step propagator estimate
\(
\left\|U_{n+\frac12}\right\|
\le 1+C_U\tau .
\) Hence
\begin{equation}
\|\widehat e_{n+\frac12}^\sigma\|
\le (1+C_U\tau)\|e_n^\sigma\|.
\end{equation}

We next include the configurations generated by the insertion.
Every configuration in \(\mathcal D_{n+1}^{[2]}\) after the insertion stage
is either an old retained configuration, a single insertion
\(\eta\cup(t_{n+\frac12})\), or a double insertion
\(\eta\cup(t_{n+\frac12},t_{n+\frac12})\), with the resulting configuration
still belonging to \(\mathcal D_{n+1}^{[2]}\).

For an old retained configuration, the preceding estimate gives
\begin{equation}
r_0^k
\|\widehat e_{n+\frac12}^\sigma\|
\le
(1+C_U\tau)E_n,
\qquad \sigma\in\Sigma_n^{(k)}.
\label{eq:second_order_old_retained_error}
\end{equation}
For a single insertion generated from \(\eta\in\mathcal D_n^{[2]}\), the two
schemes use the same insertion rule, and hence
\[
\widehat e_{n+\frac12}^{\eta\cup(t_{n+\frac12})}
=
L_{n+\frac12}\widehat e_{n+\frac12}^{\eta}.
\]
If \(|\eta|=k-1\), then
\begin{equation}
\begin{aligned}
r_0^k
\left\|
\widehat e_{n+\frac12}^{\eta\cup(t_{n+\frac12})}
\right\|
&\le
r_0^k C_L
\|\widehat e_{n+\frac12}^{\eta}\|
\\
&=
r_0C_L\,
r_0^{k-1}
\|\widehat e_{n+\frac12}^{\eta}\|
\\
&\le
(1+C_U\tau)E_n .
\end{aligned}
\label{eq:single_insertion_error}
\end{equation}
Similarly, for a double insertion,
\[
\widehat e_{n+\frac12}^{\eta\cup(t_{n+\frac12},t_{n+\frac12})}
=
L_{n+\frac12}^2\widehat e_{n+\frac12}^{\eta}.
\]
If \(|\eta|=k-2\), then
\begin{equation}
\begin{aligned}
r_0^k
\left\|
\widehat e_{n+\frac12}^{\eta\cup(t_{n+\frac12},t_{n+\frac12})}
\right\|
&\le
r_0^k C_L^2
\|\widehat e_{n+\frac12}^{\eta}\|
\\
&=
(r_0C_L)^2
r_0^{k-2}
\|\widehat e_{n+\frac12}^{\eta}\|
\\
&\le
(1+C_U\tau)E_n .
\end{aligned}
\label{eq:double_insertion_error}
\end{equation}
Combining \eqref{eq:second_order_old_retained_error}, \eqref{eq:single_insertion_error} and \eqref{eq:double_insertion_error}, we conclude that after the insertion step,
\[
\widehat E_{n+\frac12}
:=
\max_{k\ge0} r_0^k
\max_{\sigma\in\Sigma_{n+1}^{(k)}}
\|\widehat e_{n+\frac12}^\sigma\|
\le
(1+C_U\tau)E_n .
\]

During the pairing step, for $\sigma \in \mathcal D_{n+1}^{[2]}$, the retained hierarchy satisfies
\begin{equation}
    \widetilde{\psi}_{n+\frac12}^{\sigma}
    =
    \left(
    I+\tau B_{n+\frac12}^{[2]}
    \right)
    \widehat{\psi}_{n+\frac12}^{\sigma},
    \label{eq:proof_second_order_pairing}
\end{equation}
while the full hierarchy uses
\begin{equation}
    \widetilde{\phi}_{n+\frac12}^{\sigma}
    =
    \left(
    I+\tau B_{n+\frac12}
    \right)
    \widehat{\phi}_{n+\frac12}^{\sigma}.
    \label{eq:proof_second_order_pairing_full}
\end{equation}
Subtracting \eqref{eq:proof_second_order_pairing_full} from
\eqref{eq:proof_second_order_pairing} gives
\begin{equation}
\begin{aligned}
\widetilde e_{n+\frac12}^{\sigma}
&=
\widehat e_{n+\frac12}^{\sigma}
+
\tau B_{n+\frac12}^{[2]}
\widehat e_{n+\frac12}^{\sigma}
-
\tau \Delta B_{n+\frac12}
\widehat\phi_{n+\frac12}^{\sigma},
\qquad
\sigma\in\mathcal D_{n+1}^{[2]},
\end{aligned}
\label{eq:proof_second_order_pairing_error}
\end{equation}
where $\widetilde{e}_{n+\frac12}^{\sigma} = \widetilde\psi_{n+\frac12}^{\sigma} - \widetilde\phi_{n+\frac12}^{\sigma}$.

For $\sigma\in\Sigma_{n+1}^{(k)}$, the retained pairing term in \eqref{eq:proof_second_order_pairing_error} satisfies
\begin{equation}
\begin{aligned}
r_0^k
\|\tau B^{[2]}_{n+\frac12}\widehat e_{n+\frac12}^\sigma\|
&\le
C_L\tau^2 r_0^k
\sum_{j=0}^{n}
|\alpha_{n+\frac12,j+\frac12}|
\max_{\eta\in\Sigma_{n+1}^{(k+1)}}
\|\widehat e_{n+\frac12}^\eta\|  \\
&\le
\frac{C_LC_{\alpha,T}}{r_0}\tau\,
\widehat E_{n+\frac12}.
\end{aligned}
\label{eq:retained_pairing_error}
\end{equation}

Next, we estimate the defect term involving \(\Delta B_{n+\frac12}\) in
\eqref{eq:proof_second_order_pairing_error}. Let
\(\sigma\in\mathcal C_{n+1}^{[2]}\) be distinct. For every
\(0\le j\le n\), the configuration
\(\sigma\cup(t_{j+\frac12})\) is either still distinct, if
\(t_{j+\frac12}\notin\sigma\), or contains exactly one repeated label, if
\(t_{j+\frac12}\in\sigma\). In both cases,
\[
    \sigma\cup(t_{j+\frac12})\in\mathcal D_{n+1}^{[2]}.
\]
Therefore, the full and retained memory operators coincide on the distinct
configuration set:
\begin{equation}
\Delta B_{n+\frac12}\widehat\phi_{n+\frac12}^\sigma=0,
\qquad
\sigma\in\mathcal C_{n+1}^{[2]}.
\label{eq:proof_second_order_defect_vanishes_distinct}
\end{equation}

If \(\sigma\in\mathcal D_{n+1}^{[2]}\setminus\mathcal C_{n+1}^{[2]}\)
and \(|\sigma|=k\), then \(\sigma\) already contains one repeated label.
Therefore, an omitted contribution can occur only when the added label is
one of the time labels already appearing in \(\sigma\). There are at most
\(k\) such labels. Hence, by the uniform boundedness of \(L\) and \(\alpha\),
together with the level-wise bound,
\begin{equation}
\begin{aligned}
\left\|\tau\Delta B_{n+\frac12}\widehat\phi_{n+\frac12}^\sigma\right\|
&\le C_L\tau^2\sum_{\sigma\cup(t_{j+\frac12})\notin\mathcal D_{n+1}^{[2]}}|\alpha_{n+\frac12,j+\frac12}|\left\|\widehat\phi_{n+\frac12}^{\sigma\cup(t_{j+\frac12})}\right\|
\\
&\le C_L C_\alpha C_\phi \tau^2 k M^{k+1}.
\end{aligned}
\label{eq:proof_second_order_defect_intermediate_bound}
\end{equation}
Multiplying \eqref{eq:proof_second_order_defect_intermediate_bound} by \(r_0^k\), we obtain
\[
\begin{aligned}
r_0^k
\left\|
\tau\Delta B_{n+\frac12}
\widehat\phi_{n+\frac12}^\sigma
\right\|
&\le
C_L C_\alpha C_\phi \tau^2 k r_0^k M^{k+1}  \\
&=
C_L C_\alpha C_\phi M \tau^2
\, k(r_0M)^k .
\end{aligned}
\]
Since \(r_0M<1\),
\[
    k(r_0M)^k
    \le
    \sum_{\ell=1}^{\infty}
    \ell(r_0M)^\ell
    =
    \frac{r_0M}{(1-r_0M)^2}.
\]
Hence
\begin{equation}
r_0^k
\left\|
\tau\Delta B_{n+\frac12}
\widehat\phi_{n+\frac12}^\sigma
\right\|
\le
C_1\tau^2,\qquad \sigma\in\mathcal D_{n+1}^{[2]}\setminus\mathcal C_{n+1}^{[2]},
\label{eq:proof_second_order_defect_bound}
\end{equation}
where
\[
    C_1
    :=
    C_L C_\alpha C_\phi M
    \frac{r_0M}{(1-r_0M)^2}.
\]
Combining \eqref{eq:proof_second_order_defect_bound} with
\eqref{eq:proof_second_order_defect_vanishes_distinct}, we obtain
\begin{equation}
\max_{k\ge0}
r_0^k
\max_{\sigma\in\Sigma_{n+1}^{(k)}}
\left\|
\tau\Delta B_{n+\frac12}
\widehat\phi_{n+\frac12}^{\sigma}
\right\|
\le
C_1\tau^2 .
\label{eq:proof_second_order_defect_total_bound}
\end{equation}

After the left half-step,
\[
    e_{n+1}^\sigma
    =
    U_{n+\frac12}\widetilde e_{n+\frac12}^\sigma .
\]
Using the propagator bound, the estimate \eqref{eq:retained_pairing_error} for the retained memory term, and the defect bound \eqref{eq:proof_second_order_defect_total_bound},
we obtain
\[
\begin{aligned}
E_{n+1}
&\le
(1+C_U\tau)
\left[
\widehat E_{n+\frac12}
+
\frac{C_LC_{\alpha,T}}{r_0}\tau
\widehat E_{n+\frac12}
+
C_1\tau^2
\right] \\
&\le
(1+C_U\tau)
\left(1+\frac{C_LC_{\alpha,T}}{r_0}\tau\right)
\widehat E_{n+\frac12}
+
(1+C_U\tau)C_1\tau^2 \\
&\le
(1+C_U\tau)^2
\left(1+\frac{C_LC_{\alpha,T}}{r_0}\tau\right)
E_n
+
(1+C_U\tau)C_1\tau^2 .
\end{aligned}
\]
For \(0<\tau\le \tau_0\), 
\[
(1+C_U\tau)^2
\left(1+\frac{C_LC_{\alpha,T}}{r_0}\tau\right)
\le
1+C_2\tau,
\qquad
(1+C_U\tau)C_1\tau^2
\le
C_3\tau^2 .
\]
Therefore
\begin{equation}
    E_{n+1}\le (1+C_2\tau)E_n+C_3\tau^2 .
    \label{eq:second_order_E_gronwall}
\end{equation}
Since $E_0=0$, for $0\le n\le N$ with $N\tau\le T$,
\begin{equation}
E_n
\le
C_3\tau^2\sum_{\ell=0}^{n-1}(1+C_2\tau)^\ell
\le
C_E\tau,
\qquad C_E := C_3Te^{C_2T}.
\label{eq:second_order_E_bound}
\end{equation}

We next sharpen the estimate on the distinct configurations $\mathcal{C}_{n}^{[2]} = \cup_{k\ge0}\Gamma_n^{(k)}$. Define
\begin{equation}
\widehat F_{n+\frac12}
:=
\max_{k\ge0} r_1^k
\max_{\sigma\in\Gamma_{n+1}^{(k)}}
\|\widehat e_{n+\frac12}^\sigma\|.
\end{equation}
Since the two schemes use the same right half-step and insertion rules,
and since \(r_1C_L\le r_0C_L\le1\), the insertion stage does not enlarge the
weighted error. Hence
\begin{equation}
    \widehat F_{n+\frac12}
    \le
    (1+C_U\tau)F_n .
    \label{eq:proof_second_order_distinct_half_step_insertion_bound}
\end{equation}

As shown in \eqref{eq:proof_second_order_defect_vanishes_distinct}, the defect term
vanishes on distinct configurations. Consequently, for
\(\sigma\in\Gamma_{n+1}^{(k)}\), the pairing error satisfies
\begin{equation}
\widetilde e_{n+\frac12}^{\sigma}
=
\widehat e_{n+\frac12}^{\sigma}
+
\tau B_{n+\frac12}^{[2]}
\widehat e_{n+\frac12}^{\sigma}.
\end{equation}
Equivalently,
\begin{equation}
\begin{aligned}
\widetilde e_{n+\frac12}^{\sigma}
&=
\widehat e_{n+\frac12}^{\sigma}
-
\tau^2 L_{n+\frac12}^\dagger
\sum_{\substack{0\le j\le n}}
\omega_j
\alpha_{n+\frac12,j+\frac12}
\widehat e_{n+\frac12}^{\sigma\cup(t_{j+\frac12})},
\end{aligned}
\end{equation}
where \(\omega_j=1\) for \(0\le j\le n-1\) and \(\omega_n=\frac12\).

For \(\sigma\in\Gamma_{n+1}^{(k)}\), one has
\(\sigma\cup(t_{j+\frac12})\in \mathcal D_{n+1}^{[2]}\) for every
\(0\le j\le n\). Thus,
\begin{equation}
\begin{aligned}
r_1^k
\|\widetilde e_{n+\frac12}^{\sigma}\|
&\le
r_1^k
\|\widehat e_{n+\frac12}^{\sigma}\|
+
r_1^k C_L\tau^2
\sum_{j=0}^{n}
\omega_j
|\alpha_{n+\frac12,j+\frac12}|
\left\|
\widehat e_{n+\frac12}^{\sigma\cup(t_{j+\frac12})}
\right\|
\\
&\le
\widehat F_{n+\frac12}
+
r_1^k C_L\tau^2
\sum_{\substack{0\le j\le n\\ t_{j+\frac12}\notin\sigma}}
|\alpha_{n+\frac12,j+\frac12}|
\left\|
\widehat e_{n+\frac12}^{\sigma\cup(t_{j+\frac12})}
\right\|
\\
&\quad
+
r_1^k C_L\tau^2
\sum_{\substack{0\le j\le n\\ t_{j+\frac12}\in\sigma}}
|\alpha_{n+\frac12,j+\frac12}|
\left\|
\widehat e_{n+\frac12}^{\sigma\cup(t_{j+\frac12})}
\right\|
\\
&\le
\widehat F_{n+\frac12}
+
\frac{C_LC_{\alpha,T}}{r_1}\tau\widehat F_{n+\frac12}
+
C_LC_\alpha\tau^2 k r_1^k r_0^{-(k+1)}
\widehat E_{n+\frac12}.
\end{aligned}
\label{eq:proof_second_order_distinct_memory_error}
\end{equation}

Taking the maximum in \eqref{eq:proof_second_order_distinct_memory_error}
over \(\sigma\in\Gamma_{n+1}^{(k)}\) and then over \(k\ge0\), we obtain
\begin{equation}
\widetilde F_{n+\frac12}
\le
\left(
1+\frac{C_LC_{\alpha,T}}{r_1}\tau
\right)
\widehat F_{n+\frac12}
+
C_LC_\alpha\tau^2
\sup_{k\ge1}
k r_1^k r_0^{-(k+1)}
\widehat E_{n+\frac12},
\label{eq:second_order_distinct_memory_error_intermediate}
\end{equation}
where
\[
\widetilde F_{n+\frac12}
:=
\max_{k\ge0} r_1^k
\max_{\sigma\in\Gamma_{n+1}^{(k)}}
\|\widetilde e_{n+\frac12}^\sigma\|.
\]
Since \(r_1<r_0\),
\[
\sup_{k\ge1}
k r_1^k r_0^{-(k+1)}
=
\frac1{r_0}
\sup_{k\ge1}
k\left(\frac{r_1}{r_0}\right)^k
<\infty .
\]
Thus, \eqref{eq:second_order_distinct_memory_error_intermediate} can be
simplified as
\begin{equation}
\widetilde F_{n+\frac12}
\le
\left(
1+\frac{C_LC_{\alpha,T}}{r_1}\tau
\right)
\widehat F_{n+\frac12}
+
C_4\tau^2\widehat E_{n+\frac12}, \quad C_4:=C_LC_\alpha r_0^{-1}\sup_{k\ge1}
k(r_1/r_0)^k.
\label{eq:second_order_distinct_memory_error_bound}
\end{equation}
Using the previous bound \eqref{eq:proof_second_order_distinct_half_step_insertion_bound} and the estimate \eqref{eq:second_order_E_bound} of $E_n$, \eqref{eq:second_order_distinct_memory_error_bound} implies 
\[
\widetilde F_{n+\frac12}
\le
\left(
1+\frac{C_LC_{\alpha,T}}{r_1}\tau
\right)
(1+C_U\tau)F_n
+
C_4C_E\tau^3 .
\]
After applying the left half-step propagation,
\[
F_{n+1}
\le
(1+C_U\tau)\widetilde F_{n+\frac12}.
\]
Therefore, for \(0<\tau\le\tau_0\),
\begin{equation}
\begin{aligned}
F_{n+1}
&\le (1+C_U\tau)
\left(
1+\frac{C_LC_{\alpha,T}}{r_1}\tau
\right)
(1+C_U\tau)F_n
+(1+C_U\tau)C_4C_E\tau^3 \\
&\le(1+C_5\tau)F_n+C_6\tau^3 ,
\end{aligned}
\label{eq:second_order_F_gronwall}
\end{equation}
where \(C_5\) and \(C_6\), independent of \(\tau\), are chosen such that
\[
(1+C_U\tau)^2
\left(
1+\frac{C_LC_{\alpha,T}}{r_1}\tau
\right)
\le
1+C_5\tau,
\qquad
(1+C_U\tau)C_4C_E
\le C_6.
\]

Since \(F_0=0\), applying the discrete Gr\"onwall inequality to \eqref{eq:second_order_F_gronwall} gives
\begin{equation}
\begin{aligned}
F_n
&\le
C_6\tau^3\sum_{\ell=0}^{n-1}(1+C_5\tau)^\ell  \\
&\le
C_6\tau^3 n e^{C_5 n\tau}  \\
&\le
C_6Te^{C_5T}\tau^2,
\qquad 0\le n\le N,\quad N\tau\le T.
\end{aligned}
\end{equation}
Set \(C_T:=C_6Te^{C_5T}\). Since
\(\emptyset\in\Gamma_n^{(0)}\), we have
\[
\|\psi_n^\emptyset-\phi_n^\emptyset\|
\le F_n
\le C_T\tau^2 .
\]
Taking the maximum over \(0\le n\le N\) gives
\begin{equation}
\max_{0\le n\le N}
\|\psi_n^\emptyset-\phi_n^\emptyset\|
\le C_T\tau^2,
\qquad N\tau\le T.
\end{equation}
The constant \(C_T\) is independent of \(\tau\). This proves the claim.
\end{proof}

\section{Example of the Second-Order Scheme}
\label{sec:second_order_example}
This section provides an explicit illustration of one step of the second-order hierarchical update. The example demonstrates how the propagation, insertion, and pairing operations combine to generate all admissible auxiliary states at $t_1$, including both single- and double-insertion contributions. The construction follows the update rules in Section~\ref{sec:methods} and gives a concrete realization of the retained hierarchy.
\begin{align*}
    |\psi^{\emptyset, \emptyset}_1\rangle &= 
    \adjustbox{}{\includegraphics{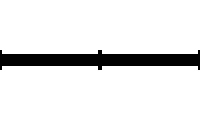}} \notag \\
    &= \adjustbox{}{\includegraphics{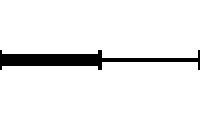}} +
       \adjustbox{}{\includegraphics{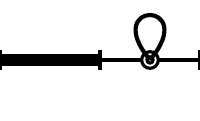}} + \adjustbox{}{\includegraphics{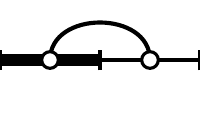}} \\
    |\psi^{\{t_\frac{1}{2}\}, \emptyset}_1\rangle &= 
    \adjustbox{}{\includegraphics{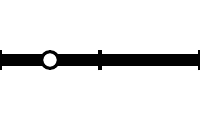}} = \adjustbox{}{\includegraphics{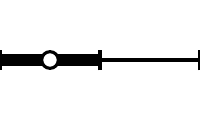}} +
       \adjustbox{}{\includegraphics{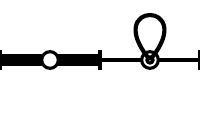}}  \\
    |\psi^{\{t_{\frac{3}{2}}\}, \emptyset}_1\rangle &= 
    \adjustbox{}{\includegraphics{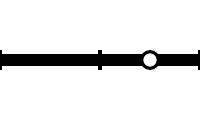}} = \adjustbox{}{\includegraphics{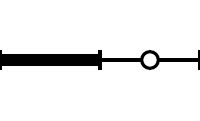}} +
       \adjustbox{}{\includegraphics{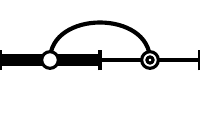}}  \\
    |\psi^{\{t_\frac{1}{2}, t_{\frac{3}{2}}\}, \emptyset}_1\rangle &= 
    \adjustbox{}{\includegraphics{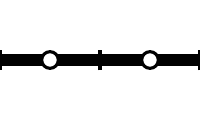}} = \adjustbox{}{\includegraphics{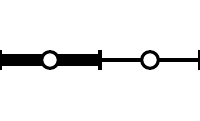}} \\
    |\psi^{\{t_\frac{1}{2}\}, t_\frac{1}{2}}_1\rangle &= 
    \adjustbox{}{\includegraphics{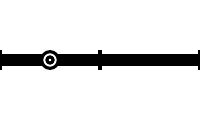}} = \adjustbox{}{\includegraphics{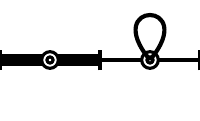}} +
       \adjustbox{}{\includegraphics{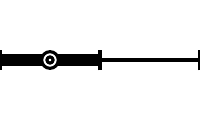}} \\
    |\psi^{\{t_{\frac{3}{2}}\}, t_{\frac{3}{2}}}_1\rangle &= 
    \adjustbox{}{\includegraphics{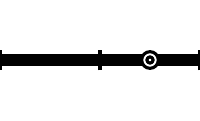}} 
    = \adjustbox{}{\includegraphics{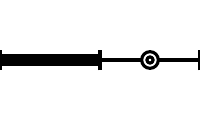}} \\
    |\psi^{\{t_\frac{1}{2},t_{\frac{3}{2}}\}, t_{\frac{3}{2}}}_1\rangle &= 
    \adjustbox{}{\includegraphics{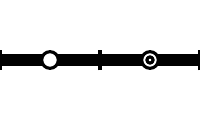}} 
    = \adjustbox{}{\includegraphics{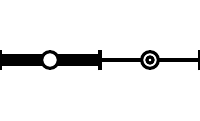}} \\
    |\psi^{\{t_\frac{1}{2},t_{\frac{3}{2}}\}, t_\frac{1}{2}}_1\rangle &= 
    \adjustbox{}{\includegraphics{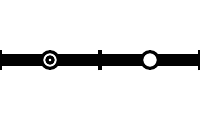}} 
    = \adjustbox{}{\includegraphics{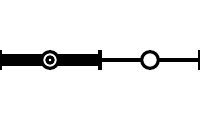}} 
\end{align*}

\section{Numerical Settings and Code Availability}
\label{sec:supp-experimental-settings}

A reference implementation of the first- and second-order NMQSD schemes is
publicly available at
\url{https://github.com/vickor-z/Numerical-Solver-for-NMQSD}.
The repository contains the core hierarchical solvers and examples.
This section records the numerical settings used for the experiments in the main manuscript.
\subsection{Common Implementation Details}

All numerical experiments were performed in double precision using only CPU
computation on a machine equipped with an AMD Ryzen AI 9 HX 365 processor and
\(32\,\mathrm{GB}\) of memory. No GPU acceleration was used. The implementation
is written in Python and uses NumPy and SciPy for numerical linear algebra,
while Numba is used to accelerate performance-critical parts of the trajectory
propagation. The stochastic trajectories were processed in batches to limit
memory consumption, and fixed random seeds were used to ensure reproducibility.

All stochastic simulations use independent realizations of the complex Gaussian
process \(z_t\). Since \(z_t\) is determined entirely by the bath correlation
function \(\alpha(t,s)\) and is independent of the system state, each complete
noise trajectory is generated in advance and then supplied to the time-stepping
scheme.

For a temporal grid \(\{t_i\}_{i=0}^{N}\), the bath correlation function is
discretized to form the covariance matrix
\[
    \Sigma_{ij}=\alpha(t_i,t_j).
\]
A Cholesky factorization
\(
    \Sigma=CC^\dagger \)
is then computed, and the correlated complex Gaussian vector is generated by
\[
    \boldsymbol{z}=C\boldsymbol{\xi},
\]
where the entries of \(\boldsymbol{\xi}\) are independent standard complex
Gaussian random variables satisfying
\[
    \mathbb{E}[\xi_i]=0,
    \qquad
    \mathbb{E}[\xi_i\xi_j]=0,
    \qquad
    \mathbb{E}[\xi_i\xi_j^*]=\delta_{ij}.
\]
The noise is sampled on a sufficiently fine temporal grid, typically with a
resolution of \(\tau/4\) or smaller. For comparisons involving different time
steps, the trajectories are generated on a common fine grid, and the samples
required by each coarser discretization are obtained by subsampling the same
realizations. This procedure preserves the prescribed discrete covariance and
reduces sampling variability in comparisons between different discretizations.

Each exponential stochastic propagator is evaluated using the second-order
Taylor approximation
\[
    \exp(hA)
    \approx
    I+hA+\frac{h^2}{2}A^2,
\]
where \(h\) denotes the corresponding propagation substep.

\subsection{Experiment-specific settings}

The numerical settings used in the main manuscript are summarized in
Tables~\ref{tab:supp-exponential-settings}--\ref{tab:supp-multilevel-settings}.
 In the Ohmic spin--boson examples, the energy scale is fixed by setting
\(\Delta=1\).

\begin{table}[htbp]
\centering
\small
\renewcommand{\arraystretch}{1.2}
\caption{Numerical settings for Example~5.1 with the exponential bath
correlation
\(
\alpha(t,s)=\frac12 e^{-|t-s|}
\).
}
\label{tab:supp-exponential-settings}
\begin{tabular}{
    p{0.3\textwidth}
    p{0.6\textwidth}}
\toprule
Setting & Value \\
\midrule

System &
\(
H=\frac12\sigma_z,
\qquad
L=\sqrt{2}\sigma_z
\)
\\

Initial state &
\(
|\psi_0\rangle
=
\frac{1+2i}{\sqrt{7}}|\uparrow\rangle
+
\frac{1+i}{\sqrt{7}}|\downarrow\rangle
\)
\\

Final time and time steps &
\(
T=2,
\qquad
\tau\in\{0.5,0.4,0.2,0.1\},
\qquad
\tau_{\rm base}=0.025
\)
\\

Convergence test &
\(
N_{\rm traj}=10^6
\)
\\

Trajectory-number test &
\(
\tau=0.1,
\qquad
N_{\rm traj}\in\{10,10^3,10^5\}
\)
\\

Truncation &
\(
K_T=1,
\qquad
M=2
\)
\\

Observables &
\(
\langle\sigma_x\rangle_t,
\qquad
\langle\sigma_y\rangle_t,
\qquad
\langle\sigma_z\rangle_t
\)
\\

Analytical reference &
\(
F(t)=it+4(e^{-t}+t-1)
\)
\\

\bottomrule
\end{tabular}
\end{table}

Examples~5.2--5.4 use an Ohmic bath with spectral density
\[
    J(\omega)
    =
    \frac{\pi}{2}\xi\omega e^{-\omega/\omega_c}.
\]
The continuous spectrum is represented by \(N_\omega=400\) modes with
\(\omega_{\max}=4\omega_c\). The remaining bath parameters are listed
separately for each experiment. Example~5.4 uses the same Ohmic-bath discretization but applies it to a
multi-level system with different physical and numerical parameters.

\begin{table}[htbp]
\centering
\small
\renewcommand{\arraystretch}{1.2}
\caption{Numerical settings for the Ohmic spin--boson experiments in
Examples~5.2 and~5.3. }
\label{tab:supp-ohmic-spin-boson-settings}
\begin{tabular}{
    p{0.25\textwidth}
    p{0.32\textwidth}
    p{0.32\textwidth}}
\toprule
Setting & Example~5.2 & Example~5.3 \\
\midrule

Hamiltonian &
\(
H=\epsilon\sigma_z+\Delta\sigma_x
\)
&
\(
H=\Delta\sigma_x
\)
\\

Coupling operator &
\(
L=\sigma_z
\)
&
\(
L=\sigma_z
\)
\\

Bias &
\(
\epsilon\in\{0,\Delta,2\Delta\}
\)
&
\(
\epsilon=0
\)
\\
Bath parameters & $\omega_c=2.5, \ \beta = 5\Delta^{-1}$ & $\omega_c=2.5, \ \beta = 5\Delta^{-1}$ \\

Coupling strengths &
\(
\xi=0.2
\)
&
\(
\xi\in\{0.2,0.4,0.6,0.8\}
\)
\\

Initial state &
\(
|\psi_0\rangle=(1,0)^\top
\)
&
\(
|\psi_0\rangle=(1,0)^\top
\)
\\

Observable &
\(
\langle\sigma_z\rangle_t
\)
&
\(
\langle\sigma_z\rangle_t
\)
\\

Final time &
\(
T=5
\)
&
\(
T=5
\)
\\

Time step &
\(
    \tau\in\{0.4,0.2,0.1,0.05,0.025\}
\)
&
\(
\tau=0.1
\)
\\

Number of trajectories &
\(
N_{\rm traj}=10^5
\)
&
\(
N_{\rm traj}=10^4
\)
\\

Memory truncation &
\(
K_T=1
\)
&
$K\in \{9,12,15,18,21,24,27\}$
\\

Hierarchy truncation &
\(
M=2
\)
&
\(
M\in\{1,2,3,4,5\}
\)
\\

\bottomrule
\end{tabular}
\end{table}

\begin{table}[htbp]
\centering
\small
\renewcommand{\arraystretch}{1.2}
\caption{Numerical settings for the \(d=11\) multi-level system in
Example~5.4 with an Ohmic bath.}
\label{tab:supp-multilevel-settings}
\begin{tabular}{
    p{0.3\textwidth}
    p{0.63\textwidth}}
\toprule
Setting & Value \\
\midrule

System dimension &
\(
d=11
\)
\\

Hamiltonian &
\(
H
=
\sum_{i=1}^{d-1}
\left(
|i\rangle\langle i+1|
+
|i+1\rangle\langle i|
\right)
\)
\\

Coupling operator &
\(
L
=
\operatorname{diag}
\left(
-1,-1+\frac{2}{d-1},\ldots,1
\right)
\)
\\

Bath parameters &
\(
\omega_c=2.5,
\qquad
\xi=0.2,
\qquad
\beta=\Delta^{-1}
\)
\\

Initial states &
\(
\rho_s(0)=|1\rangle\langle1|
\)\ 
and
\ \(
\rho_s(0)
=
\frac12
\left(
|1\rangle\langle1|
+
|d\rangle\langle d|
\right)
\)
\\

Final time and time step &
\(
T=10,
\qquad
\tau=0.1
\)
\\

Number of trajectories &
\(
N_{\rm traj}=10^4
\)
\\

Truncation &
\(
K=18,
\qquad
M=3
\)
\\

Observable &
\(
P_i(t)=\langle i|\rho_s(t)|i\rangle,
\quad i=1,\ldots,d
\)
\\

\bottomrule
\end{tabular}
\end{table}
 
\end{document}